\documentclass{article}
\usepackage{natbib}
\usepackage[dvips]{graphicx}
\usepackage{latexsym}
\usepackage{amsthm}
\usepackage{mathrsfs}
\usepackage[usenames]{color}
\usepackage{psfrag}

\usepackage{amsbsy,amsfonts,amssymb,amsmath}

\usepackage{natbib}

\newcommand{\boldA}{\mathbf{A}}
\newcommand{\boldB}{\mathbf{B}}

\DeclareMathOperator{\matexpb}{E} 
\DeclareMathOperator{\esb}{E_{\textrm{\tiny S}}} 
\DeclareMathOperator{\evb}{E_{\textrm{\tiny V}}} 
\DeclareMathOperator{\edab}{E_{\textrm{\tiny DA}}}

\newcommand{\matexp}[1]{\matexpb[#1]}
\newcommand{\es}[1]{\esb[#1]}
\newcommand{\ev}[1]{\evb[#1]}
\newcommand{\eda}[1]{\edab[#1]}

\DeclareMathOperator{\cov}{cov} 

\newcommand{\unif}[2]{\textrm{U}[#1, #2]}

\newcommand{\mtranspose}{\mathrm{\textrm{\tiny T}}}

\newcommand{\ud}{\mathrm{d}}

\newcommand{\ind}[1]{\mathbb{I}({#1})}

\newcommand*{\dd}{\; \mbox{d}}
\newcommand{\interior}[1]{\textrm{int\;}#1}
\newcommand{\cl}[1]{\textrm{cl\;}#1}
\newcommand{\RofN}[1]{\mathbb{R}^#1}

\newcommand{\uchi}{\textrm{\raisebox{0.2ex}{$\chi$}}}
\newcommand{\charomega}[1]{\mbox{\large $\uchi$}_{\mbox{\tiny $#1$}}}
\newcommand{\charx}[2]{\mbox{\large $\uchi$}_{\mbox{\tiny $#1$}}(#2)}
\newcommand{\hatchi}[1]{\mbox{\large $\widehat \uchi$}_{\mbox{\tiny $#1$}}}
\newcommand{\hatchix}[2]{\mbox{\large $\widehat \uchi$}_{\mbox{\tiny $#1$}}(#2)}

\newcommand{\calP}{\mathscr{A}}
\newcommand{\calF}{\mathcal{F}}
\newcommand{\calK}{\mathcal{K}}
\newcommand{\calD}{\mathcal{D}}

\newcommand{\textg}{\textrm{g}}

\newcommand{\Lp}[1]{L_{#1}}
\newcommand{\bP}{\textrm{\bf P}}

\newcommand{\calm}{\mathfrak{m}}
\newcommand{\cald}{\mathfrak{d}}
\newcommand{\calW}{\mathcal{W}}
\newcommand{\calmw}{\mathfrak{m}_{\calW}}
\newcommand{\calY}{\mathcal{Y}}
\newcommand{\calH}{\mathcal{H}}
\newcommand{\parA}{\mathfrak{A}}

\DeclareMathOperator*{\arginf}{arg\,inf}

\DeclareMathOperator*{\argmin}{arg\,min}

\newtheorem{thm}{\noindent Theorem}[section]
\newtheorem{cor}{\noindent Corollary}[section]
\newtheorem{prop}{\noindent Proposition}[section]

\theoremstyle{definition}
\newtheorem{ex}{\noindent Example}[section]
\newtheorem{defn}{\noindent Definition}[section]
\newtheorem{rem}{\noindent Remark}[section]

\title{Expectations of Random Sets and Their Boundaries \\ Using Oriented Distance Functions}

\author{Larissa I. Stanberry$^{*}$ and Hanna K. Jankowski$^{**}$ \vspace{5mm} \\ 
$^{*}$School of Mathematics, University of Bristol. \vspace{1mm} \\
$^{**}$Department of Mathematics and Statistics, York University.}

\begin{document}
\maketitle

\begin{abstract}
Shape estimation and object reconstruction are common problems in image analysis. Mathematically, viewing objects in the image plane as random sets reduces the problem of shape estimation to inference about sets. Currently existing definitions of the expected set rely on different criteria to construct the expectation. This paper introduces new definitions of the expected set and the expected boundary, based on oriented distance functions. The proposed expectations have a number of attractive properties, including inclusion relations, convexity preservation and equivariance with respect to rigid motions. The paper introduces a special class of separable oriented distance functions for parametric sets and gives the definition and properties of separable random closed sets. Further, the definitions of the empirical mean set and the empirical mean boundary are proposed and empirical evidence of the consistency of the boundary estimator is presented. In addition, the paper gives loss functions for set inference in frequentist framework and shows how some of the existing expectations arise naturally as optimal estimators. The proposed definitions of the set and boundary expectations are illustrated on theoretical examples and real data.
\end{abstract}

\section{Introduction}
\label{sec:intro}

Boundary reconstruction and shape estimation are frequently encountered problems in image analysis. For example, it is often of interest to determine a characteristic shape of a cell, reconstruct tissue boundary in medical images, or estimate a probable area in the earthquake disaster zone. Mathematically, the objects of interest can be viewed as sets, whilst inherent stochasticity of the acquisition process turns them into random entities. The problems of boundary reconstruction and shape estimation thus reduce to inference about random sets. 

Early results in the theory of random sets date back to \citet{Choquet1953} with a thorough mathematical treatment appearing in \citet{Matheron1975} and the most recent developments in the field presented in \citep{Molchanov2005}. Because the family of closed sets is nonlinear, there is no natural way to define the expected set. Currently, there exist a number of definitions of the expectation, though neither can be spoken of as the best. 
As pointed out in \citet{Molchanov2005}, the definition of the expectation depends on set features that are important to emphasize. 

Many of the existing definitions are not based on a random set directly, but use an embedding into a function space. For example, the Vorob'ev definition uses the representation of a set given by its characteristic function to construct an expectation that is optimal with respect to Lebesgue measure \citep{Vorobev1984}. The distance-average approach maps a set into the space of distance functions, so that the expected set is optimal among all of the level sets of the expected distance function with respect to a predefined metric \citep{BaddeleyMolchanov1998}. 
 
Currently existing definitions of the expected set do not address the question of boundary estimation, which is often of primary interest in practice. In this paper, we introduce a new definition of the expected set and the expected boundary based on oriented distance functions (ODFs). As opposed to a conventional distance function, an ODF takes into account the set and its complement, providing a more informative representation of a set \citep{DelfourZolesio2001}. In comparison with currently existing definitions, the new expectations have a number of attractive theoretical properties. In particular, they satisfy inclusion relations, preserve convexity and remain equivariant with respect to rigid motions. We introduce a new class of ODFs for parametric sets and show the connection between the expectation of a parametric random closed set and the expected parameter. We also outline a general framework for set inference and show how different expectations arise as natural estimators for different loss functions.


\section{Random Closed Sets}
\label{sec:rcs}

\subsection{Notation}
For set $A \subset \RofN{d}$, denote by $A^c, \partial A, \cl{A}, \interior{A}, \lambda({A})$ its complement, boundary, closure, interior, and Lebesgue measure, respectively. We write $x=(x_1, \ldots, x_d)$ for point $x \in \RofN{d}$ and use $|\cdotp|$ for the standard Euclidean norm. 
Let $\calF$ (resp. $\calK$) be the family of closed (resp. compact) subsets of $\RofN{d}$ and let the tripple $(\Omega,\calP, \bP)$ denote the probability space.

\subsection{Definitions and Examples}
\begin{defn}\label{def:racs}
A random closed set is the mapping $\boldA : \Omega \mapsto \calF$ such that for every compact set $K\in \calK$, $$\lbrace\omega: \boldA(\omega) \cap K\neq\emptyset\rbrace \in \calP.$$
\end{defn}
Here, we consider random closed sets taking values in $\RofN{d}$. Throughout this paper, we write ``r.c.s.'' for a random closed set, though some texts prefer RACS. We write $\boldA$ for an r.c.s. $\boldA(\omega)$ and reserve capital roman letters for closed subsets of $\RofN{d}$.

Random closed sets give rise to several random variables and random functions, including the characteristic function 
\begin{equation}\label{def:charf}
      \charomega{A}(x) =\left\lbrace \begin{array}{ll}
        1, & \textrm{ if $x\in A$},\\
        0, & \textrm{ if $x\notin A$},
        \end{array}
        \right.
\end{equation}
the Lebesgue measure $\lambda(A) = \int_{\RofN{d}}\charomega{A}(x)\ud x$,
and the distance function from point $x\in \RofN{d}$ to $A$,
    \begin{equation}\label{def:dist}
    d_{A}(x)=\left\lbrace \begin{array}{ll}
        \inf_{y \in A}\rho(x,y), & A \neq \emptyset,\\
        +\infty, & A = \emptyset,
    \end{array} \right.
    \end{equation}
where $\rho(x,y)$ is a metric on $\RofN{d}$. Some examples of $\rho$ include the $L_{2}$ Euclidean metric, the $L_{\infty}$ Chebyshev metric, and the $L_{1}$ Manhattan metric. Throughout this paper, we use the Euclidean distance function with $\rho(x,y) = |x-y|$ in \eqref{def:dist}.

\subsection{Expectations of Random Closed Sets}

Because the space $\calF$ of all closed subsets of $\RofN{d}$ is nonlinear, there is no natural way to define the expected set. In this section, we briefly review some of the existing definitions and refer for more details to \citep{Aumann1965, ArtsteinVitale1975, Vorobev1984, BaddeleyMolchanov1998, Molchanov2005}.

\begin{defn}[Selection expectation]\label{def:selection}
A random element $\xi$ is called the selection of $\boldA$, if it belongs to $\boldA$
with probability one. A selection $\xi$ is integrable, if $\matexpb{|\xi|}<\infty$.
The selection expectation, $\es{\boldA}$,  is the closure of the
set of the expectations of all integrable selections of $\boldA$,
$$\es{\boldA} =\cl{\lbrace \matexp{\xi} \colon \xi \textrm{ is a selection and } \matexpb{|\xi|}<\infty \rbrace}.$$
\end{defn}

\noindent The selection expectation depends on the structure of the probability space \citep[see Example 1.14, Section 2]{Molchanov2005}. In addition, if the probability space is nonatomic, $\es\boldA$ is necessarily convex, even for nonconvex deterministic sets.

In the context of image analysis, the Vorob'ev expectation is perhaps the most intuitive construction \citep{Vorobev1984}. Consider an r.c.s. $\boldA\subset \calD$. For every point $x\in \calD$, we assign a probability mass depending on whether or not $x\in\boldA$. For example, assigning one to all points in $\boldA$ and zero to points in $\boldA^{c}$ converts the observed set into a binary image. Averaging binary images over all realizations of $\boldA$, we obtain a gray-scale image, where the intensity at each pixel is the probability of the point being in $\boldA$. The gray-scale image is not a binary image, unless $\boldA$ is deterministic. The Vorob'ev definition then provides a criterion to construct an expected set that is optimal with respect to Lebesgue measure.

More precisely, given the characteristic function $\charomega{\boldA}$ of an r.c.s. $\boldA$, define the coverage function by
\begin{equation}\label{eq:covFn}
\matexp{\charx{\boldA}{x}}=\bP(x\in \boldA).
\end{equation}
The excursions sets of the coverage function are given by
\begin{equation}\label{eq:exSets}
  \boldA_u=\lbrace x\in \RofN{d}: \matexp{\charx{\boldA}{x}} \geq u\rbrace, \quad q\in [0,1].
\end{equation}
\begin{defn}[Vorob'ev expectation]\label{def:vorExp}
The expectation, $\ev{\boldA}$, of a random closed set $\boldA$ is the excursion set $\boldA_q$, where $q\in[0,1]$ is such that
\begin{equation}\label{def:vor}
\lambda(\boldA_u)\leq \matexp{\lambda(\boldA)}\leq \lambda(\boldA_q) \quad \textrm{for all } u>q.
\end{equation}
\end{defn}
\noindent Note that the optimality criterion \eqref{def:vor} ensures that the Vorob'ev expectation ignores sets of measure zero.

The distance-average expectation is based on a representation of a set given by some function $f_{A}\colon \calD\mapsto \mathbb{R}$ \citep{BaddeleyMolchanov1998}. We call $f_{A}$ a representative function of the set $A$. Examples of $f$ include the distance function, the ODF, the characteristic function and others. For an r.c.s. $\boldA$ with representative function $f_{\boldA}$, let $\matexp{f_{\boldA}(x)}$ be the expected value of $f_{A}$ at $x$, assuming it exists. 
\begin{defn}[Distance-average expectation]\label{def:da}
For a compact window $\calW$, define a (pseudo-) metric $\calm(\boldA, \boldB) = \calm_{\calW}(f_{\boldA}(\cdot), f_{\boldB}(\cdot))$. The distance-average expectation, $\eda{\boldA}$, of an r.c.s. $\boldA$ is the level set 
\begin{equation}\label{def:daExp}
\boldA_{u} = \lbrace x \in \calW \colon \matexp{f_{\boldA}(x)}\leq  u \rbrace, \textrm{ where } u = \arginf_{s \in \mathbb{R}} \calm (\matexp{f_{\boldA}},  f_{\boldA_{s}})
\end{equation}
\end{defn}
Examples of pseudometric $\calm$ in \eqref{def:daExp} include $\Lp{q}$ distances and their variates, Baddeley's $\Delta^{q}$-distance, or any custom definitions \citep{Baddeley1992}. Note that the pseudometric $\calm$ is computed over the window $\calW$, which can be the entire domain $\calD$ or its subset. For brevity, we omit the subscript $\calW$ whenever possible. 

The distance-average expectation strongly depends on the choice of the representative function $f$, window $\calW$, (pseudo-) metric $\calm$ and parameters of $\calm$ \citep{BaddeleyMolchanov1998}. The Vorob'ev expectation can be viewed as a special case of the distance-average definition with $f_{\boldA}(x)=1-\chi_{\boldA}(x)$ and $L^{1}$ distance $\calm$ \citep[see Example~5.14]{BaddeleyMolchanov1998}.

The linearization approach gives a unified view of the existing expectations \citep{Molchanov2005}. In particular, consider a mapping $f \colon \calF \mapsto \calY$, where $\calY$ is some Banach space. Let $f_{\boldA}$ be the image of $\boldA$ in $\calY$ and assume that $\matexp{|f_{\boldA}(x)|}<\infty$. If there exists a unique $F \in \calF$ such that $f_{F}(x) = \matexp{f_{\boldA}(x)}$, then declare $\matexp{\boldA} = F$. In general, the inverse image rarely exists, and we need to define the criterion so that the resulting expectation is optimal in some sense. This is done as follows. 
Let $\cald$ be a pseudometric in $\calY$, then the expectation of an r.c.s. $\boldA$ is given by 
\begin{equation}\label{def:linAppr}
\matexp{\boldA} = \argmin_{F\in\calF}\cald(\matexp{\zeta(\boldA)}, \zeta(F)).
\end{equation}
Minimizing over the entire family $\calF$ is often nontrivial. Therefore, we can consider optimizing $\cald(\cdot, \cdot)$ over a subfamily $\calH\subset\calF$ of candidate sets. Clearly, the resulting expectation depends on the choice of Banach space $\calY$, mapping $f$, metric $\cald$, subfamily $\calH$ and any implicit parameters that enter into the calculation. 

For bounded closed covex sets, the selection expectation is an example of a linearization approach, where the embedding into a Banach space is realised by support functions of random sets on the unit sphere and the expected set is given by the support function of the selection expectation. Another example is the Vorob'ev definition, which is based on the mapping given by $\charomega{\boldA}$ and gives the expected set $\ev{\boldA}$ which is optimal with respect to Lebesgue measure; the family of candidate sets $\calH$ consists of excursion sets \eqref{eq:exSets} of the coverage function. The distance-average expectation is based on the embedding into the space of representative function and the resulting average is optimal among the level sets of the expected representation with respect to metric $\calm(\cdot, \cdot)$ \citep{BaddeleyMolchanov1998, Molchanov2005}.

\subsection{Oriented Distance Functions}
\label{sec:ODF}

In comparison to a distance function (\ref{def:dist}), an ODF takes into account both the set and its compliment, thus reflecting the interior and exterior properties. In addition, the ODF provides a level-set description of the boundary.
\begin{defn}\label{def:sdist}
The ODF from point $x$ to set $A\subset \RofN{d}$ with $\partial A \neq \emptyset$ is defined by
\begin{equation}\label{eq:sdist}
  b_A(x)=d_A(x)-d_{A^c}(x), \ \ \textrm{for all } x\in \RofN{d}.
\end{equation}
\end{defn}
\noindent Using basic properties of the distance function, the expression (\ref{eq:sdist}) can be written as
\begin{equation}
  b_A(x)=\left\lbrace\begin{array}{ll}
  \hspace{3mm} d_A(x)= d_{\partial A}(x), & x\in \interior{A^c} \\
  \hspace{3mm} 0, & x\in \partial A\\
  -d_{A^c}(x)=-d_{\partial A}(x), & x \in \interior{A} \;.
  \end{array} \right.
\end{equation}
Below we summarize some properties of the ODFs and refer for proofs and details to \citet{DelfourZolesio2001}.

\begin{prop}
Let $A$ and $B$ be some subsets of $\cl \calD$ with $\partial A, \partial B \neq \emptyset$.
\begin{enumerate}
\item \label{odf:prop1} The ODF provides a level-set description of a set, i.e 
$$\cl{A}=\lbrace x: b_A(x)\leq 0 \rbrace , \quad \partial A=\lbrace x \colon b_A(x)=0 \rbrace .$$
\item\label{odf:prop2} $A\supset B$ implies $b_A\leq b_B$.
\item\label{odf:prop3} $b_A\leq b_B$ in $\calD$ iff $\cl{B}\subset \cl{A} \textrm{ and } \cl{A^{c}}\subset \cl{B^c}$.
\item\label{odf:prop4}  \label{property:sdist1:iii}$b_A=b_B$ in $\calD$ iff $\cl{A}=\cl{B}$ and $\partial A=\partial B$.
\item\label{odf:prop5} For a convex set $A$ with $\partial A \neq \emptyset,\; b_A=b_{\cl{A}}$ is a convex function in $\RofN{d}$.
\item\label{odf:prop6}  The function $b_A(x)$ is uniformly Lipschitz in $\RofN{d}$ with constant one, i.e.
\begin{equation*}
  |b_A(y)-b_A(x)|\leq |y-x| \quad \textrm{for any $x,y\in \RofN{d}$}.
\end{equation*}
In addition, $b_A(x)$ is (Fre\'chet) differentiable with $|\nabla b_A(x)|\leq 1$ a.e. in $\RofN{d}.$ When it exists, the gradient of the ODF coincides with the outward unit normal to the boundary.
\end{enumerate}
\end{prop}

Note that the ODF identifies the set up to its closure and boundary. In other words, ODF describes equivalence classes of sets with identical closure and boundary. 



\begin{ex}\label{example:ODFpoint}
For a singleton $A=\lbrace \theta \rbrace, \; \theta\in \RofN{d}$, the ODF coincides with the distance function and $b_A(x)=d_A(x)=|x-\theta|$ for all $x\in\RofN{d}$. Note that $b_{A}(x)>0$ for all $x\neq \theta$ and the isocontours of $b_{A}(x)$ are spheres centered at $\theta$; Figure \ref{fig:exODFs}, 1st plot.
\end{ex}

\begin{ex}\label{example:ODFdisc}
Consider a closed ball in $\RofN{d}$ centered at the origin with radius $\theta$, i.e. $A=\lbrace x \in \RofN{d}\colon|x|\leq \theta \rbrace$. The distance function of $A$ is $d_A(x)= |x|-\theta$, for $x\in A$, and zero, elsewhere. Hence, the ODF has a form $b_{A}(x)=|x|-\theta$ for all $x\in \RofN{d}.$ The isocontours of $b_{A}(x)$ are spheres centered at the origin; Figure \ref{fig:exODFs}, 2nd plot.
\end{ex}

\begin{ex}\label{example:ODFhplane}
The distance function of a half plane set  $A=\lbrace x \colon x_1 \leq \theta_1 \rbrace$ is given by $d_A(x)=x_1-\theta_1$ for $x\in A^{c}$. The ODF has a form $b_A=x_1-\theta_1$ for all $x \in \RofN{2}$ with isocontours given by vertical lines parallel to the boundary; Figure \ref{fig:exODFs}, 3rd plot.
\end{ex}

\begin{ex}\label{example:ODFtplane}
The ODF of an upper-plane set $ A=\lbrace x \colon x_2 \geq  x_1 \tan{\theta}\rbrace$ in $\RofN{2}$ is $b_{A}(x) = |x| \sin{(\theta-\omega)}$, where $\omega=\arcsin{(x_{2}|x|^{-1})}$; Figure \ref{fig:exODFs}, 4th plot.
\end{ex}

\begin{figure}[!hbtp]
{\scriptsize
\psfrag{-2.5}[][]{} \psfrag{-1.5}[][]{} \psfrag{-0.5}[][]{} \psfrag{0.5}[][]{} \psfrag{1.5}[][]{} \psfrag{2.5}[][]{} 
\psfrag{-2}[][]{-2} \psfrag{-1}[][]{-1} \psfrag{1}[][]{ 1} \psfrag{2}[][]{ 2} \psfrag{0}[][]{ 0} 
\centerline{
     \includegraphics[scale=.3]{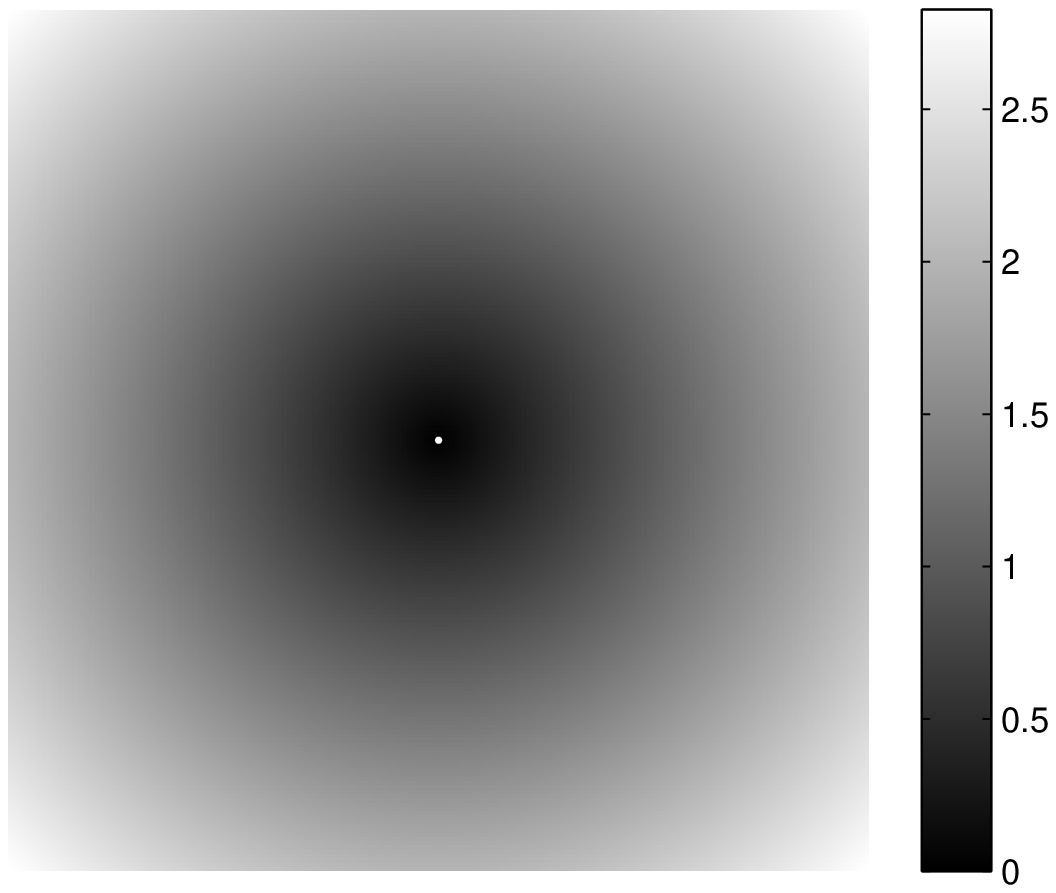}
     \includegraphics[scale=.3]{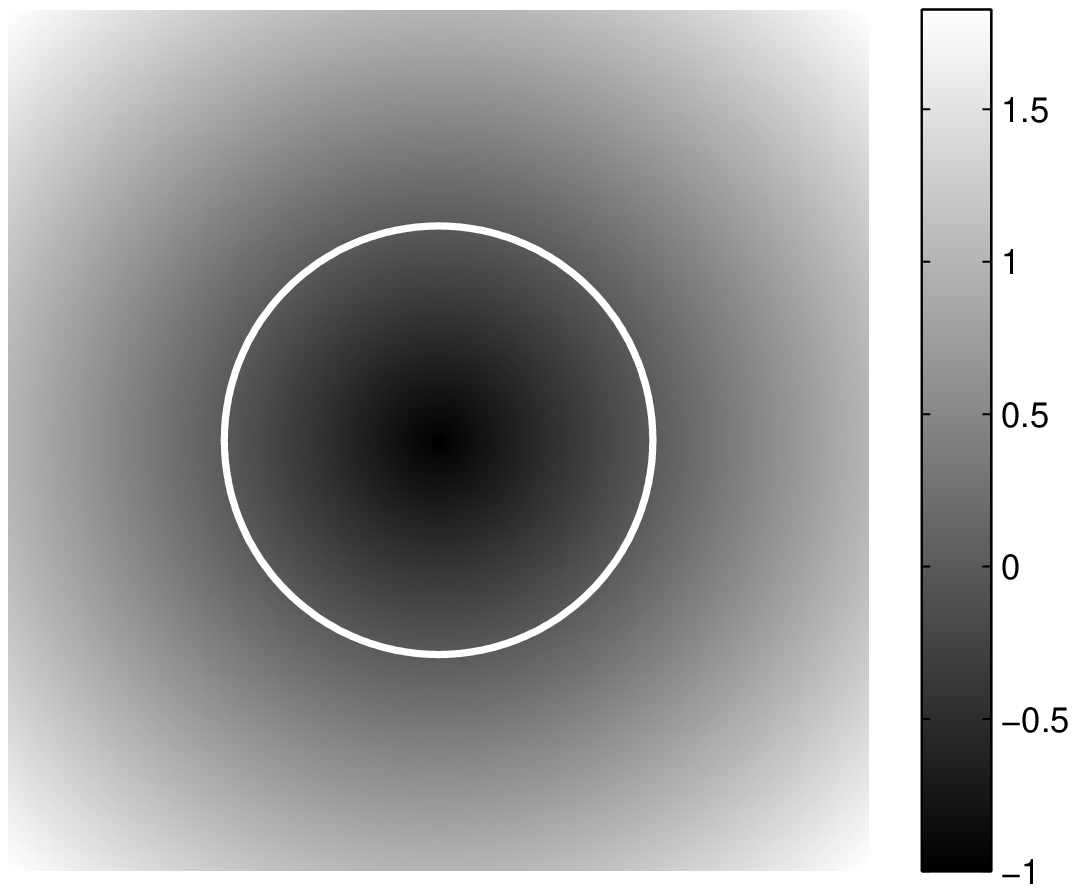}
  	\includegraphics[scale= .3]{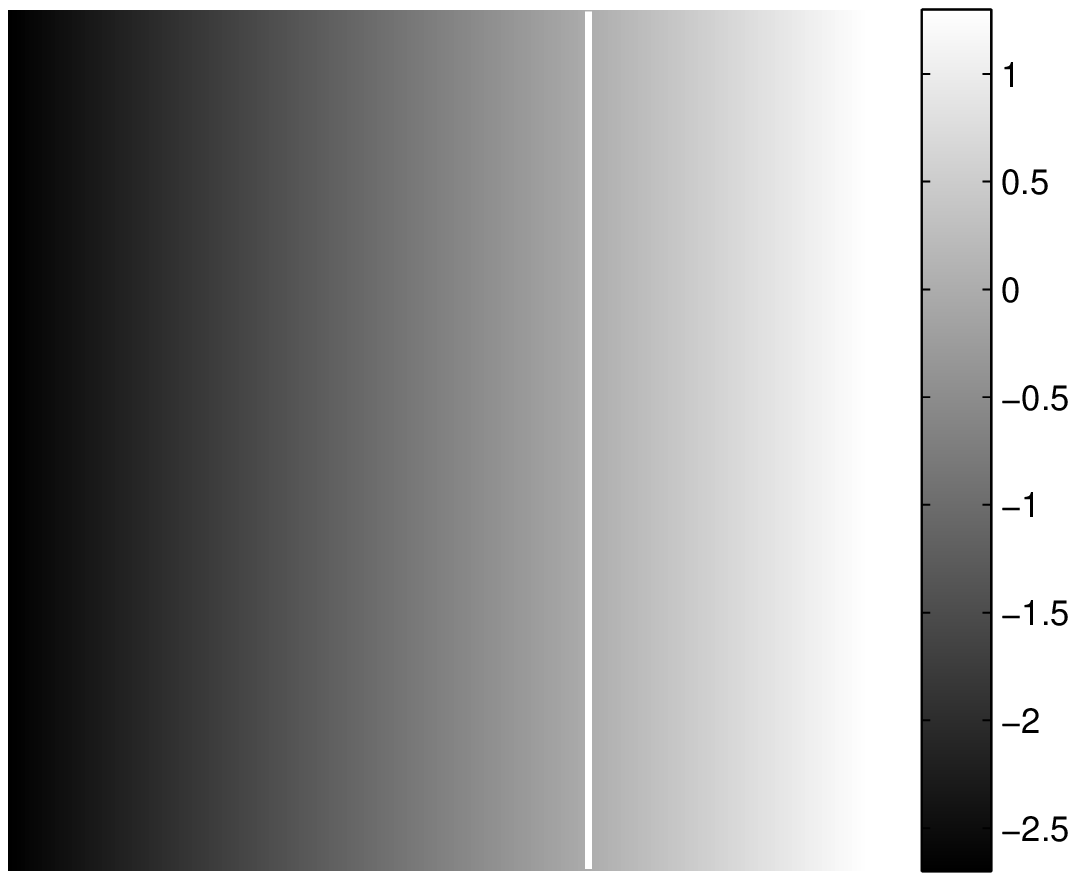}
	\includegraphics[scale= .3]{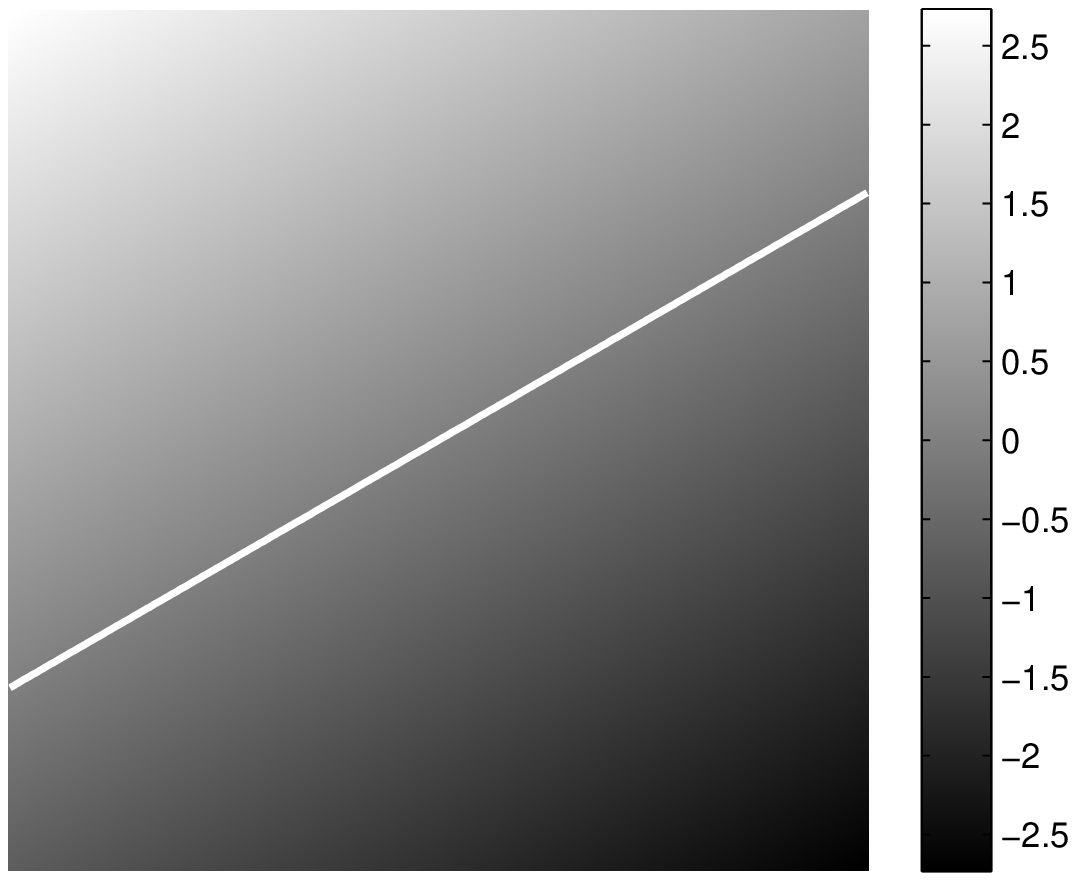}    
}}
\caption{From left to right: The ODF of a singleton $\lbrace \theta \rbrace$, a unit disc, a half-plane and an upper half-plane in $\RofN{2}$; the zero-level isocontour (white) shows the boundary of the set.} \label{fig:exODFs}
\end{figure}

\section{Expectations of Random Sets via ODFs}
\label{sec:ODFexp}

\subsection{Definition and Properties}

For an r.c.s. $\boldA$ with $\partial \boldA \neq \emptyset$, the ODF $b_{\boldA}(x)$ at point $x\in\RofN{d}$ is a random variable taking values in $\mathbb{R}$. Denote by $\matexp{b_{\boldA}(x)}$ the expected value of $b_{A}(x)$, assuming it exists. We begin by proving some properties of the expected ODF.

\begin{prop}\label{odfLip}
Consider an r.c.s. $\boldA$ such that $\partial \boldA \neq \emptyset$ a.s. Then,
\begin{enumerate}
\item \label{prop:odfELip} The expected ODF is uniformly Lipshitz with constant 1. 
\item \label{prop:odfExp} If $\matexp{|b_{\boldA}(x_{0})|}<\infty$ for some $x_{0}\in \calD$, then $\matexp{|b_{\boldA}(x)|}<\infty$ for any $x\in \calD$. 
\item \label{prop:odfGrad} If $b_{\boldA}(x)$ is differentiable at some $x_{0}\in \calD$ a.s., then $\nabla \matexp{b_{\boldA}(x_{0})}$ exists and $|\nabla \matexp{b_{\boldA}(x_{0})}|\leq 1$. 
\end{enumerate}
\end{prop}
\begin{proof}
\eqref{prop:odfELip} For any $x,y\in \RofN{d}, \; |\matexp{b_{\boldA}(y)}- \matexp{b_{\boldA}(x)}|\leq \matexp{|b_{\boldA}(y)-b_{\boldA}(x)|}\leq |y-x|$.

\eqref{prop:odfExp} Using the triangle inequality, boundedness of $\calD$ and Lipshiptz continuity of $\matexp{b_{\boldA}(x)}$, we obtain $\matexp{|b_{\boldA}(x)|}\leq |x-x_{0}| + \matexp{|b_{\boldA}(x_{0})|}<\infty$. 

\eqref{prop:odfGrad} Because $\nabla b_{\boldA}(x_{0})$ exists a.s. and $|\nabla b_{A}(x_{0})| \leq 1$, we can interchange the differentiation on the space $\Omega$ and integration with respect to $x$, so the existence of $\nabla \matexp{b_{\boldA}(x_{0})}$ follows. From $\nabla b_{\boldA}(x_{0})\leq 1$, we have $|\nabla \matexp{b_{\boldA}(x_{0})}|\leq  \matexp{|\nabla b_{\boldA}(x_{0})|}\leq 1$.
\end{proof}

\begin{defn}[ODF expectation] \label{def:sdistMeanSet}
Let $\boldA$ be a random closed set such that $\partial \boldA\neq \emptyset$ a.s. and suppose that $\matexp{b_{\boldA}(x_{0})}<\infty$ for some $x_{0}\in\calD$. Then the expectation of an r.c.s. $\boldA$ is the zero-level set of the expected ODF and expectation of the boundary of $\boldA$ is the zero-level isocontour of the expected ODF, i.e.
\begin{eqnarray}\label{def:odfMean}
\matexp{\boldA}& = & \lbrace x:
\matexp{b_{\boldA}(x)}\leq 0\rbrace, \\
\matexp{\partial \boldA}& = & \lbrace x: \matexp{b_{\boldA}(x)}=0\rbrace.
\end{eqnarray}
\end{defn}
\noindent We refer to $\matexp{\boldA}$ as ``the expected set" or ``the ODF expectation"  of $\boldA$. Note, that the expected set defined by (\ref{def:sdistMeanSet}) can be an empty set. Later we give examples of such random sets and discuss the appropriateness of the outcome.

\begin{thm}[Inclusion properties]\label{thm1}
    Consider r.c.s. $\boldA, \mathbf{B} \subset \cl{D}$ with $\partial \boldA, \partial \boldB\neq \emptyset$ and assume that $\matexp{b_{\boldA}(x)}$ and $\matexp{b_{\boldB}(x)}$ exist for some $x \in \cl{\calD}$. That is, suppose $\matexp{\boldA}$ and $\matexp{\boldB}$ are well-defined. Then,
\begin{enumerate}
    \item\label{thm1:itemClSet} The expectation $\matexp{\boldA}$ is a closed set.
    \item\label{thm1:itemDetSet} If $\boldA = A$ a.s., for some deterministic $A\in\calF$, then $\matexp{\boldA}= A$ and $\matexp{\partial \boldA}=\partial {A}$.
        \item\label{thm1:itemInclB} Suppose that an r.c.s. $\boldA$ satisfies $\partial \boldA=B$ a.s. for some deterministic set $B$, then $\matexp{\partial \boldA} \supset B$.     
    \item\label{thm1:itemInclBb} For an r.c.s. $\boldA, \; \partial \matexp{\boldA} \subset \matexp{\partial \boldA}$.
    \item\label{thm1:itemInclODF}  If $\matexp{b_{\boldA}(x)} \geq \matexp{b_{\boldB}(x)}$ a.s., then $\matexp{\boldA} \subset \matexp{\boldB}$.
    \item\label{thm1:itemInclAB} If $\boldA \subset \boldB$ a.s., then $\matexp{\boldA} \subset \matexp{\boldB}$.
\end{enumerate}
\end{thm}
\begin{proof}
\eqref{thm1:itemClSet} The closedness of $\matexp{\boldA}$ follows from the continuity of $\matexp{b_{\boldA}}$.

\eqref{thm1:itemDetSet} is straightforward.

\eqref{thm1:itemInclB} It suffices to note that for all points $x\in B, \ b_{B}(x)=b_{\boldA}(x)=0$, so that $\matexp{b_{\boldA}(x)} = 0$ and $B\subset \matexp{\partial \boldA}$. 

\eqref{thm1:itemInclBb} For any $x\in \partial \matexp{\boldA}, \matexp{b_{\boldA}(x)}=0$ and hence $x\in \matexp{\partial \boldA}$. Example \ref{ex:setBoundary} shows that the reverse relation does not hold. 

\eqref{thm1:itemInclODF} For any $x\in \matexp{\boldA}, \ 0\geq \matexp{b_{\boldA}(x)} \geq \matexp{b_{\boldB}(x)}$. Thus $x\in \matexp{\boldB}$ and $\matexp{\boldA} \subset \matexp{\boldB}$. 

\eqref{thm1:itemInclAB} The proof follows immediately from the previous statement and Property \ref{odf:prop3} of the ODFs, which implies that for
$\boldA \subset \boldB, \; b_{\boldA}(x)\geq b_{\boldB}(x)$ in $\calD$ and hence, $\matexp{b_{\boldA}(x)}\geq \matexp{b_{\boldB}(x)}$ a.s..
\end{proof}

\begin{cor}\label{thm1:itemInclUW} If an r.c.s. $\boldA\supset U$ a.s. for some deterministic set $U\subset \cl{\calD}$, then
        $\matexp{\boldA} \supset \cl{U}.$ Similarly, if an r.c.s. $\boldA \subset W$ a.s. for some deterministic set $W\subset \cl{\calD}$,
        then $\matexp{\boldA} \subset \cl{W}$.
\end{cor}
\begin{proof} Follows from statements \eqref{thm1:itemDetSet} and \eqref{thm1:itemInclAB} of Theorem \ref{thm1}.
\end{proof}

Although the inclusion properties of Theorem \ref{thm1} and Corollary \ref{thm1:itemInclUW} seem natural, they do not necessarily hold for other definitions, including the selection expectation \citep{Molchanov2005} and the distance-average definition. The latter does not satisfy a weaker version of inclusion in the sense that the resulting expectation is not necessarily contained in the convex hull of the union of all possible realizations of a random set \citep{BaddeleyMolchanov1998}.

\begin{ex}\label{ex:rsingleton} (random singleton)
Let $\xi$ be a random variable and consider an r.c.s. $\boldA  = \lbrace \xi \rbrace$. Then, for any $x\in \calD, \; \matexp{b_{\boldA}(x)} = \matexp{|x-\xi| >0}$ and hence $\matexp{\boldA} = \emptyset$. 
\end{ex}

The last result does not agree with suggestions in \citet{Molchanov2005}, where $\matexp{\lbrace \xi \rbrace}= \lbrace \matexp{\xi}\rbrace$ is listed as one of the desirable properties of the expected set. Indeed, the selection expectation satisfies the suggested property and so does the distance-average definition, for certain choices of $\calm, f$ and $W$ \citep[see Example~5.4] {BaddeleyMolchanov1998}. However, the ODF definition of the expected set was primarily motivated by problems in image analysis, where random singletons are often regarded as noise. Example \ref{ex:rsingleton} highlights the natural denoising property of the ODF expectation in a sense that any speckles, unless observed a.s., are not present in the average reconstruction. We argue that if random singletons are the sets of primary interests, they should be viewed as random vectors rather than random sets, in which case conventional definitions of the expectation and dispersion are applicable. 

More generally, Example \ref{ex:rsingleton} can be extended to random sets with zero Lebesgue measure. In particular, for a r.c.s. $\boldA$ with $\interior{\boldA}=\emptyset$ and $\lambda(\boldA) = 0$, we have $b_{\boldA}(x)>0$ for all $x\in \boldA^{c}$ and $b_{\boldA}(x) =0$ for all $x\in \partial \boldA$. Thus, unless $x\in \partial \boldA$ almost surely, $\matexp{b_{\boldA}(x)}>0$ and the expected set $\matexp{\boldA} = \lbrace x \colon x\in \boldA \textrm{ almost surely}\rbrace$.

In the rest of this section, we discuss equivariance properties of the ODF expectation. To begin with, recall that the translation of a set $A\subset\RofN{d}$ by $a\in\RofN{d}$ is the set $a+A=\lbrace a+x \colon x\in A\rbrace$. Similarly, the homothecy of $A$ by a scalar $\alpha$ is given by
$\alpha A=\lbrace  \alpha x \colon x\in A \rbrace$. The homothecy is a dilation for $\alpha\geq 1$, a contraction for $\alpha\in [0,1)$, and a reflection for $\alpha =-1$. For negative $\alpha$, the homothecy is a dilation or a contraction of the reflection for $|\alpha|<-1$ and $|\alpha|\in (0,1)$, respectively.
Rigid motion transformations $\textg \colon \RofN{d}\mapsto\RofN{d}$ are given by $\textg (x)=\Gamma x+a$, where $\Gamma \in O(n)$ is an orthogonal matrix and $a \in \RofN{d}$. In $\RofN{d}$, these transformations are isometries,  i.e. $|x-y| = |\textg (x)-\textg (y)|$, and form the Euclidean group $E(d)$ with respect to the composition. 

\begin{thm} [Equivariance properties] \label{thm:equiv} The expected set and the expected boundary are equivariant with respect to: 
\begin{enumerate}
\item \label{prop:dilation} Homothecy, i.e. for a fixed scalar $\alpha\neq 0$,
\begin{equation}\label{dilation:set}
\matexp{\alpha \boldA}=\alpha \matexp{\boldA} \quad \textrm{and} \quad \matexp{\partial\alpha \boldA}=\alpha \matexp{\partial \boldA}.
\end{equation}
\item \label{prop:rigid} The group of rigid motions, i.e. for any $\textg\in E(d)$,
\begin{equation}\label{rigid:set}
\matexp{\textg \boldA}=\textg \matexp{\boldA} \quad \textrm{and} \quad \matexp{\textg \; \partial \boldA}=\textg \matexp{\partial \boldA}.
\end{equation}
\end{enumerate}
\end{thm}
\begin{proof} 
\eqref{prop:dilation} We begin by establishing the relation between the ODFs of a set and its homothecy. In particular,
$d_{\alpha A}(x)=\inf_{y \in \alpha A} |y-x|=\inf_{y\in A}|\alpha y-x|=|\alpha| d_A (x/\alpha).$ Similarly,  $d_{(\alpha A)^c}(x)=|\alpha | d_{A^c}(x/\alpha)$,  so that $b_{\alpha A}(x)=|\alpha | b_A(x/\alpha).$

For any $y\in \matexp{\alpha \boldA}, \; \matexp{b_{\alpha \boldA}(y)} \leq 0$. Writing $y=\alpha x$ for some $x$, we obtain
$0\geq \matexp{b_{\alpha \boldA}(y)}=\matexp{b_{\alpha \boldA}(\alpha x)} = |\alpha |\matexp{b_{\boldA}(x)}.$
Thus, $\matexp{b_{\boldA}(x)}\leq 0$ and $y=\alpha x$ with $x\in \matexp{\boldA}$. Hence, $y\in \alpha \matexp{\boldA}$ and $\matexp{\alpha \boldA} \subseteq \alpha \matexp{\boldA}$.

For any $y\in \alpha \matexp{\boldA}, \; y=\alpha x$ for some $x\in \matexp{\boldA}$. Then $\matexp{b_{\alpha \boldA}(y)} = \matexp{|\alpha |b_{\boldA}(y/\alpha)}=|\alpha | \matexp{b_{\boldA}(x)}\leq 0$. Thus, $y \in \matexp{\alpha \boldA}$ and $ \matexp{\alpha \boldA} \supseteq \alpha \matexp{\boldA}$.

The equivariance property for the boundary is proved similarly and we omit the details. 

\eqref{prop:rigid} A transformation $\textg(x)=\Gamma x +a$, for some orthogonal matrix $\Gamma\in O(d)$ and $a\in\RofN{d}$, is a bijection and $\textg^{-1}(x)=\Gamma^{-1}(x-b)=\Gamma^{\textrm{\tiny T}}(x-b)$. For sets
$\textg \boldA = \lbrace \textg(x) \colon x\in \boldA \rbrace \quad \textrm{and} \quad (\textg \boldA)^c = \lbrace \textg(x) \colon x\in {\boldA}^c\rbrace,$
we obtain $d_{\textg A}(x) = \inf_{y\in \textg A}|y-x|=\inf_{y\in A} |\Gamma y+b-x| = \inf_{y \in A} |y-\Gamma^{\mtranspose}(x-b)|=d_A(\textg^{-1}(x))$.
Similarly, $d_{(\textg A)^c}(x)= d_{A^c} (\textg^{-1}(x)) $ and $b_{\textg A}= b_A (\textg ^{-1}(x))$.

Recall that $\textg \matexp{\boldA} = \lbrace \textg (x) \colon x\in \matexp{\boldA}\rbrace$ and $\matexp{\textg \boldA} = \lbrace x \colon \matexp{b_{\textg \boldA}(x)}\leq 0 \rbrace.$ Hence, for any $y\in \textg \matexp{\boldA}$, there exists $x\in \matexp{\boldA}$ such that $y=\textg(x)$. From here, $\matexp{b_{\textg \boldA}(y)}=\matexp{b_{\boldA}(\textg^{-1}(\Gamma x+b))}=\matexp{b_{\boldA}(x)}\leq 0.$ Hence, $y \in \matexp{\textg \boldA}$ and $\matexp{\textg \boldA} \supseteq \textg \matexp{\boldA}$.

On the other hand, any $y \in \matexp{\textg \boldA}$ can be written as $y=\Gamma x +b$, where $x=\textg^{-1}(y)$. It then follows that $0\geq \matexp{b_{\textg \boldA}(y)}=\matexp{b_{\boldA}(\textg^{-1}(\textg(x)))} = \matexp{b_{\boldA}(x)}.$ Thus, any $y \in \matexp{\textg \boldA}$ has a form $y=\Gamma x+b$, where $x\in \matexp{\boldA}$. Hence, $\matexp{\textg \boldA} \subseteq \textg \matexp{\boldA}$ and $\matexp{\textg \boldA} = \textg \matexp{\boldA}$.
\end{proof}

\begin{cor}\label{cor:transl}
Note that any translation in $\RofN{d}$ can be described as an isometry with an identity matrix $\Gamma$. Thus, it immediately follows that the expected set and the expected boundary are translation-equivariant, i.e. for some fixed $a\in \RofN{d}$,
\begin{equation}\label{prop:transl}
\matexp{a+\boldA} = a+\matexp{\boldA} \quad \textrm{and} \quad \matexp{\partial (a+\boldA)} = a+\matexp{\partial \boldA}.
\end{equation}
\end{cor}

\begin{thm}[Convexity preservation]\label{thm:convx}
Suppose that an r.c.s. $\boldA$ is convex a.s.. Then the expectation $\matexp{\boldA}$ of $\boldA$ is also convex. 
\end{thm}
\begin{proof}
The a.s. convexity of $\boldA$ implies that $\cl{\boldA}$ is convex and $b_{\boldA} = b_{\cl{\boldA}}$ is a convex function a.s. \citep[Theorem 7.1]{DelfourZolesio2001}. Hence, for all $\alpha\in[0,1], \; b_{\boldA}(\alpha x +(1-\alpha)y)\leq \alpha b_{\boldA}(x)+(1-\alpha)b_{\boldA}(y)$ a.s.. Taking expectations, we obtain that $\matexp{b_{\boldA}}$ is convex, and so is $\matexp{\boldA}$. 
\end{proof}

We refer to Theorems \ref{thm:equiv} and \ref{thm:convx} as the shape-preservation properties of the ODF expectation. These qualities are particularly desirable in shape and image analysis. In contrast, given a nonatomic probability space, the selection expectation of an r.c.s. is convex and coincides with the expectation of its convex hull. Notably, the convexification of $\es{\boldA}$ holds even for nonconvex deterministic sets \citep{Molchanov2005}. The distance-average expectation preserves the convexity, but only if the set, its representative function and the window are convex \citep{BaddeleyMolchanov1998}.  

\begin{ex} [Set and its boundary] \label{ex:setBoundary} Consider an r.c.s. $\boldA$ such that $\boldA = \lbrace 0, 1 \rbrace$ or $\boldA = [0,1]$ in $\mathbb{R}$ with probabilty $p$ and $1-p$, respectively. Note that the boundary $\partial\boldA=\lbrace 0, 1 \rbrace$ a.s. in $\mathbb{R}$, so essentially, we are observing either the set or its boundary. The expected ODF is given by 
\begin{displaymath}
\matexp{b_{\boldA}(x)} = \left\{ \begin{array}{ll}
-x, & x \leq 0 \\
(2p-1)x, & 0 < x \leq 0.5 \\
(2p-1)(1-x), & 0.5<x\leq 1 \\
x-1,  & x>1.
\end{array} \right.
\end{displaymath}
For $p=0.5, \; \matexp{\boldA} =  \matexp{\partial \boldA} =  [0,1] $. For $p<0.5, \; \matexp{\boldA} = [0,1]$ and $\matexp{\partial \boldA} = \lbrace 0, 1 \rbrace$. For $p>0.5, \; \matexp{\boldA} = \matexp{\partial \boldA} = \lbrace 0, 1 \rbrace$. Note that for $p=0.5, \; \partial \boldA = \lbrace 0,1 \rbrace$ a.s. However, $\matexp{\partial \boldA} \neq \lbrace 0, 1 \rbrace$, but rather $\partial \boldA \subset \matexp{\partial \boldA}$, as in Proposition \ref{thm1}\eqref{thm1:itemInclB}. 

For $p=0.5$, the distance-average expectation is $\eda{\boldA} = [-1/12, 1/6] \cup [5/6, 13/12]$, for $\calm$ uniform, and $\eda{\boldA} = \lbrace 0, 1 \rbrace$, for $\calm = \Delta_{\calW}^{p}$ with sufficiently large $\calW$ \citep{BaddeleyMolchanov1998}. 

More generally, consider an r.c.s. $\boldA$ such that, for some deterministic $B\in\calF, \; \boldA = B$ with probability $p$ and $\boldA = \partial B$, otherwise. From
\begin{equation*}
b_{\partial B}(x) = \left\{\begin{array}{ll}
b_{B}(x), & x\in B^{c}, \\
0, & x\in\partial B, \\
-b_{B}(x), & x\in \interior{B},
\end{array}\right.
\end{equation*}
it follows that
\begin{equation*}
\matexp{b_{\boldA}(x)} = \left\{\begin{array}{ll}
b_{B}(x), & x\in B^{c}, \\
0, & x\in\partial B, \\
(1-2p)b_{B}(x), & x\in\interior{B}.\\
\end{array}\right.
\end{equation*}
Hence, 
\begin{eqnarray*}
\textrm{for } p & = & 0.5, \quad \matexp{\boldA} = \matexp{\partial\boldA} = B, \\
\textrm{for } p & < & 0.5, \quad \matexp{\boldA} =B, \; \matexp{\partial\boldA} = \partial B, \\
\textrm{for } p & > & 0.5, \quad \matexp{\boldA}  = \matexp{\partial\boldA} = \partial B. \\
\end{eqnarray*} 
\end{ex}

\begin{ex} [Ball with random radius] \label{ex:ball2} Consider a random closed ball $\boldA$ in $\RofN{d}$ with a fixed center $x_{0}$ and a random radius $\Theta$. Recall that the ODF of $\boldA$ is given by $b_{\boldA}(x)=|x-x_{0}|-\Theta \quad \textrm{for all } x \in \RofN{d}$. Hence, $\matexp{b_{\boldA}(x; \Theta)}=|x-x_{0}|-\matexp{\Theta}$ and the expected set is a closed ball centered at $x_{0}$ with radius $\matexp{\Theta}$. The expected boundary $\matexp{\partial \boldA}$ is a sphere centered at $x_{0}$ with radius $\matexp{\Theta}$.
\end{ex}

\begin{ex}[Flashing discs] \label{example:discreteCircle} Let $\Theta$ be a Bernoulli random variable with parameter $p$. Consider an r.c.s. $\boldA \subset\RofN{2}$ given by a disc of radius $r$ centered at the origin, with probability $p$, and centered at a point $a\in \RofN{2}$, otherwise.
The ODF of $\boldA$ is given by $b_{\boldA}(x)=|x|\ind{\Theta=1}+|x-a|\ind{\Theta=0}-r,$ where $\ind{\cdot}$ is an indicator function.
Taking the expectation with respect to $\Theta$, we obtain $\matexp{b_{\boldA}(x)}=p|x|+(1-p)|x-a|-r.$ The isocontours of the expected ODF are Cartesian ovals with foci $\lbrace 0, a \rbrace$.


Figure \ref{fig:fBallsCO} shows the expected ODF with the expected boundary superimposed in white, for $p=0.8$ and three different values of $a$. The observed discs $\boldA$ are shown in gray. When the realizations of $\boldA$ do not intersect, the expected boundary is nearly a circle contained in the leftmost disc (Figure \ref{fig:fBallsCO}, left). As the distance between the centers decreases, the expected set expands. In general, as $p\rightarrow 1$, the expected ODF approaches that of a disc centered at the origin with radius $r$. Alternatively, as $p\rightarrow 1/2$, the isocontours approach an elliptical shape.

\begin{figure}[!hbtp]
{\scriptsize 
\psfrag{0.6}[t][c]{.6} \psfrag{0.8}[][]{} \psfrag{1.2}[][]{} \psfrag{}[][]{} \psfrag{1.4}[][]{} \psfrag{1.6}[][]{} \psfrag{1.8}[][]{} \psfrag{2.2}[][]{}
\psfrag{2.4}[][]{} \psfrag{2.6}[t][c]{2.6}
\psfrag{-0.5}[][]{} \psfrag{0.5}[][]{} \psfrag{1.5}[][]{} \psfrag{-1.5}[][]{} \psfrag{-2.5}[][]{} \psfrag{2.5}[][]{}
\psfrag{1}[t][c]{1} \psfrag{2}[t][c]{2} \psfrag{-1}[t][c]{-1} \psfrag{-2}[t][c]{-2} \psfrag{0}[t][c]{0} \psfrag{3}[t][c]{3} \psfrag{-3}[t][c]{-3}
\centerline{
    \includegraphics [scale=0.24]{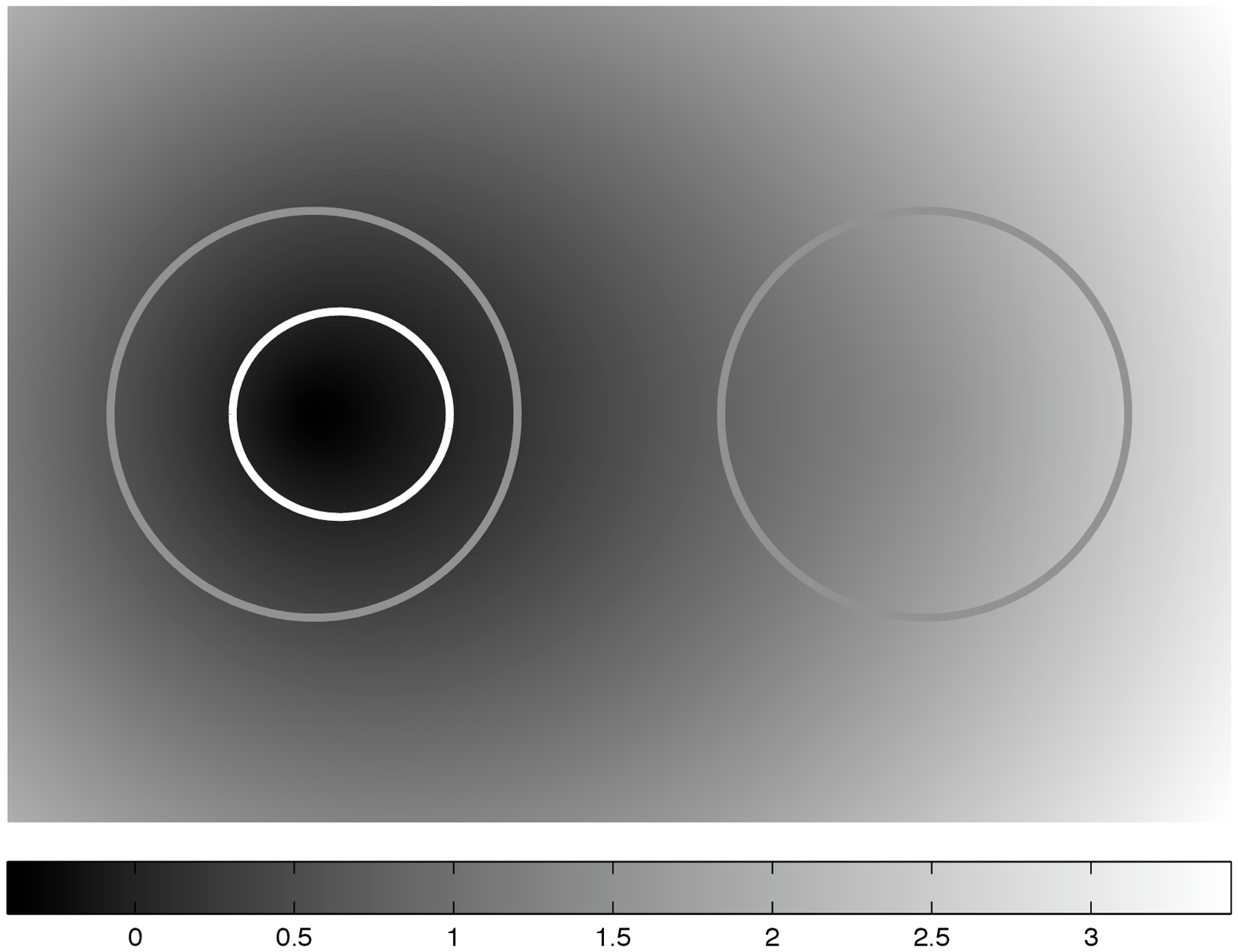} \hspace{3mm}  
    \includegraphics [scale=0.24]{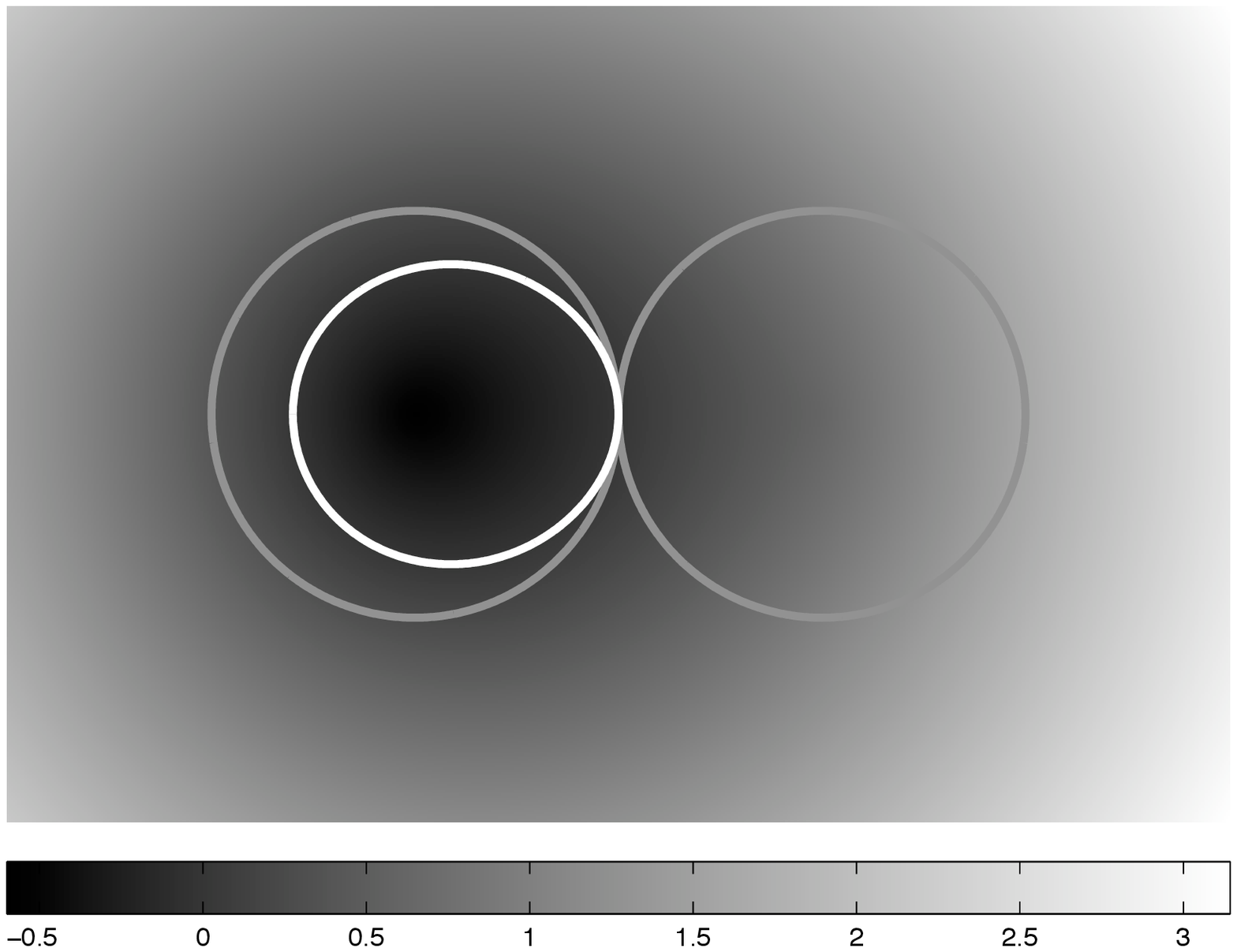} \hspace{3mm} 
    \includegraphics [scale=0.24]{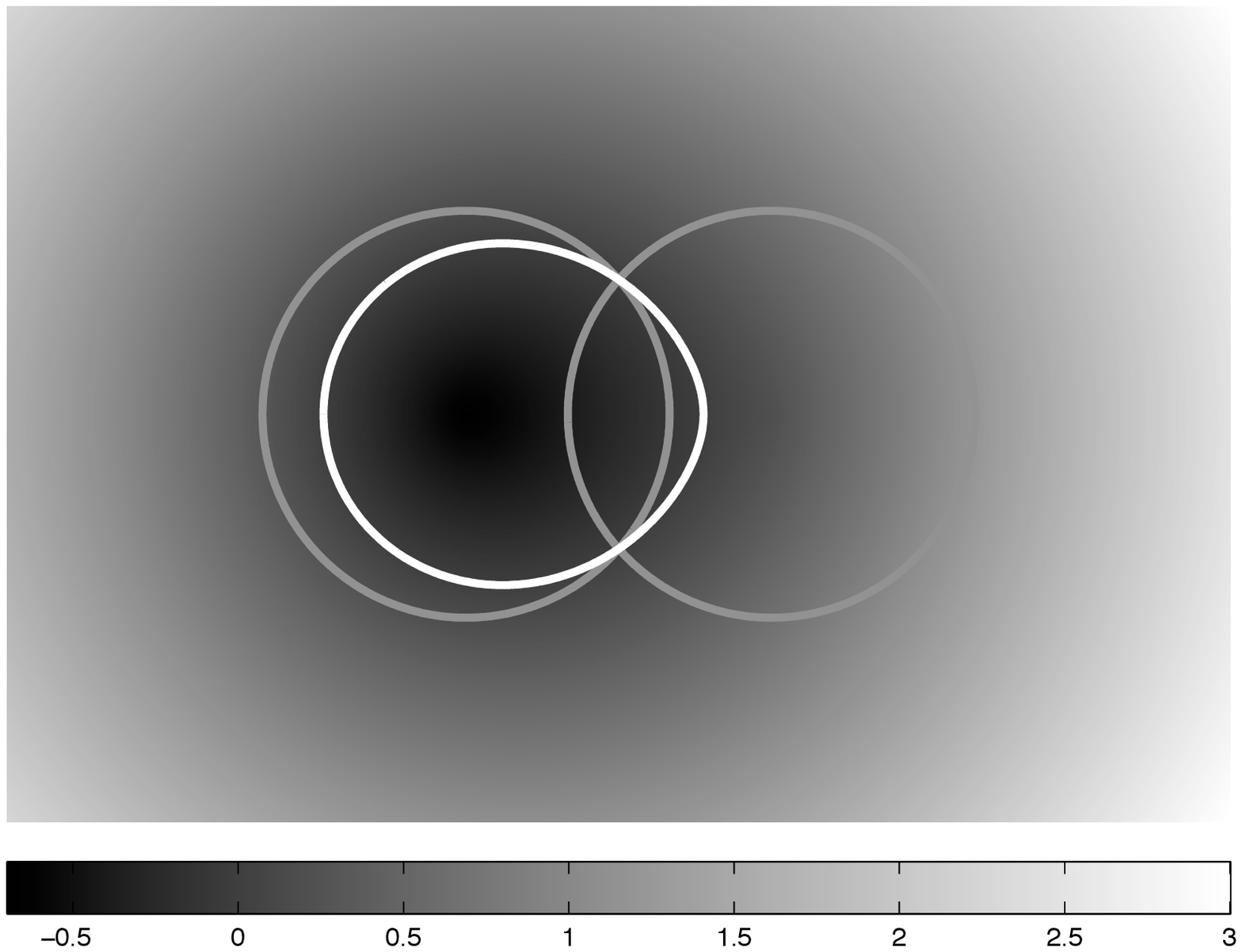} }}
\caption{The expected boundaries (white) for $\boldA$ with $p=0.8, r=1$ and $|a| = 3, 2, 1.5$ (from left to right), superimposed on the gray-scale image of the expected ODF. Realizations of $\partial\boldA$ are shown in white.} \label{fig:fBallsCO}
\end{figure}

For $p=0.5$, the $c\;$-level isocontour of the expected ODF is an ellipse with foci $(0,0)$ and $a$ and semimajor axis $r+c$. From $\matexp{b_{\boldA}(x)}\geq |a|/2-r$, we obtain that if $|a|>2r$, then $\matexp{\boldA} = \emptyset $; if $|a|=2r$, then $\matexp{\boldA} = \lbrace a/2 \rbrace$; if $|a|>2r$, then $\matexp{\boldA}$ is the closure of an ellipse with foci $\lbrace (0,0), a \rbrace $  and a semimajor axis $r$.  Figure \ref{fig:fBallsEll} shows the expected boundaries (white) overlayed on the expected ODF, for discs with $r=1$ and $p=0.5$.

In comparison, the selection expectation $\es{\boldA}$ is always a disc centered at $a/2$ with radius $r$, irrespective of $p$. The Vorob'ev expectation $\ev{\boldA}$ for $p\neq 0.5$ is given by the more likely of the two discs. For $p=0.5$, $\ev{\boldA}$ is the union of the two discs, provided the cardinality of the intersection does not exceed one, otherwise $\ev{\boldA}$ is the intersection of the two discs. 

\begin{figure}[!hbtp]
{\scriptsize 
\psfrag{0.6}[t][c]{.6} \psfrag{0.8}[][]{} \psfrag{1.2}[][]{} \psfrag{}[][]{} \psfrag{1.4}[][]{} \psfrag{1.6}[][]{} \psfrag{1.8}[][]{} \psfrag{2.2}[][]{}
\psfrag{2.4}[][]{} \psfrag{2.6}[t][c]{2.6}
\psfrag{-0.5}[][]{} \psfrag{0.5}[][]{} \psfrag{1.5}[][]{} \psfrag{-1.5}[][]{} \psfrag{-2.5}[][]{} \psfrag{2.5}[][]{}
\psfrag{1}[t][c]{1} \psfrag{2}[t][c]{2} \psfrag{-1}[t][c]{-1} \psfrag{-2}[t][c]{-2} \psfrag{0}[t][c]{0} \psfrag{3}[t][c]{3} \psfrag{-3}[t][c]{-3}
\centerline{
    \includegraphics [scale=0.24]{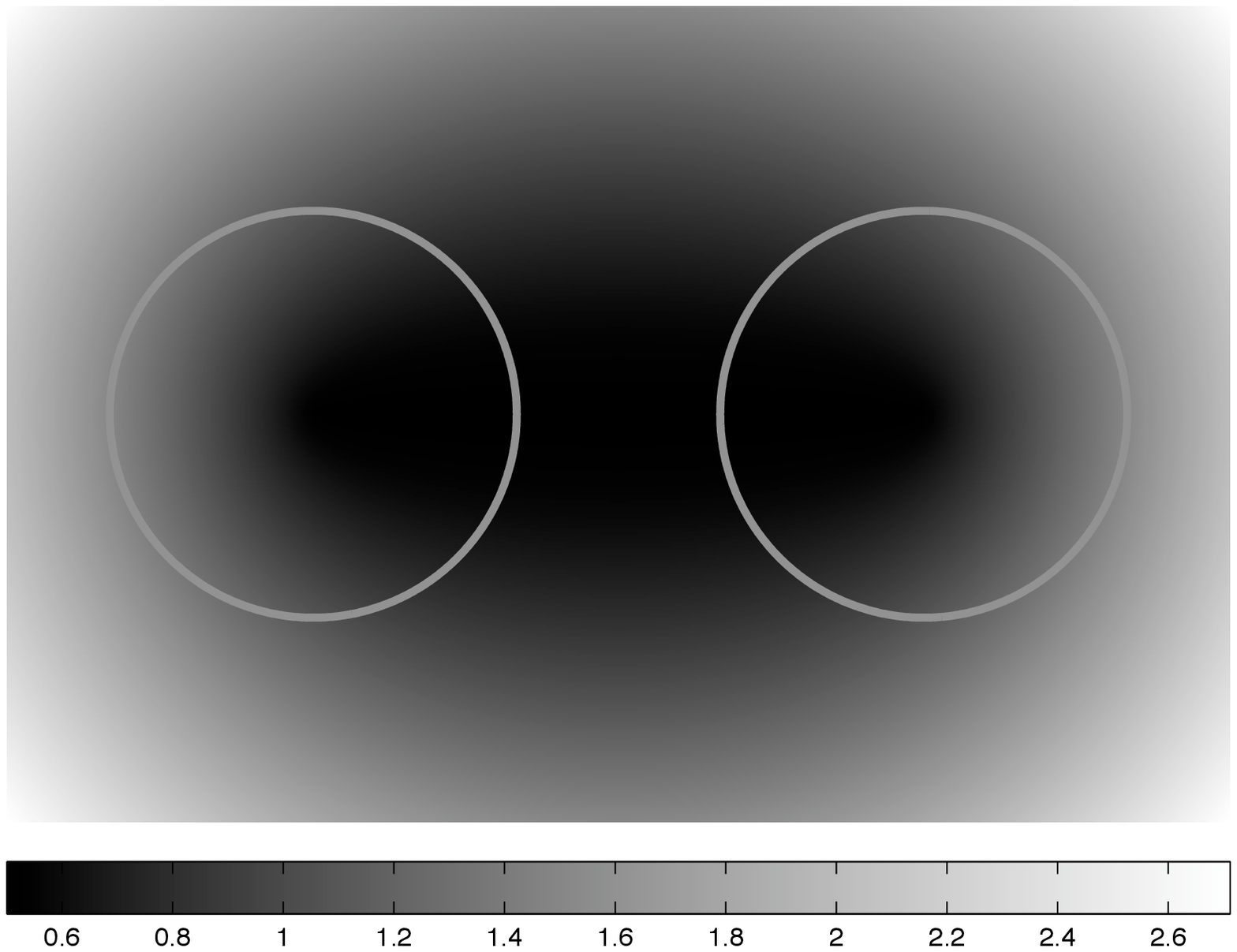} \hspace{3mm}  
    \includegraphics [scale=0.24]{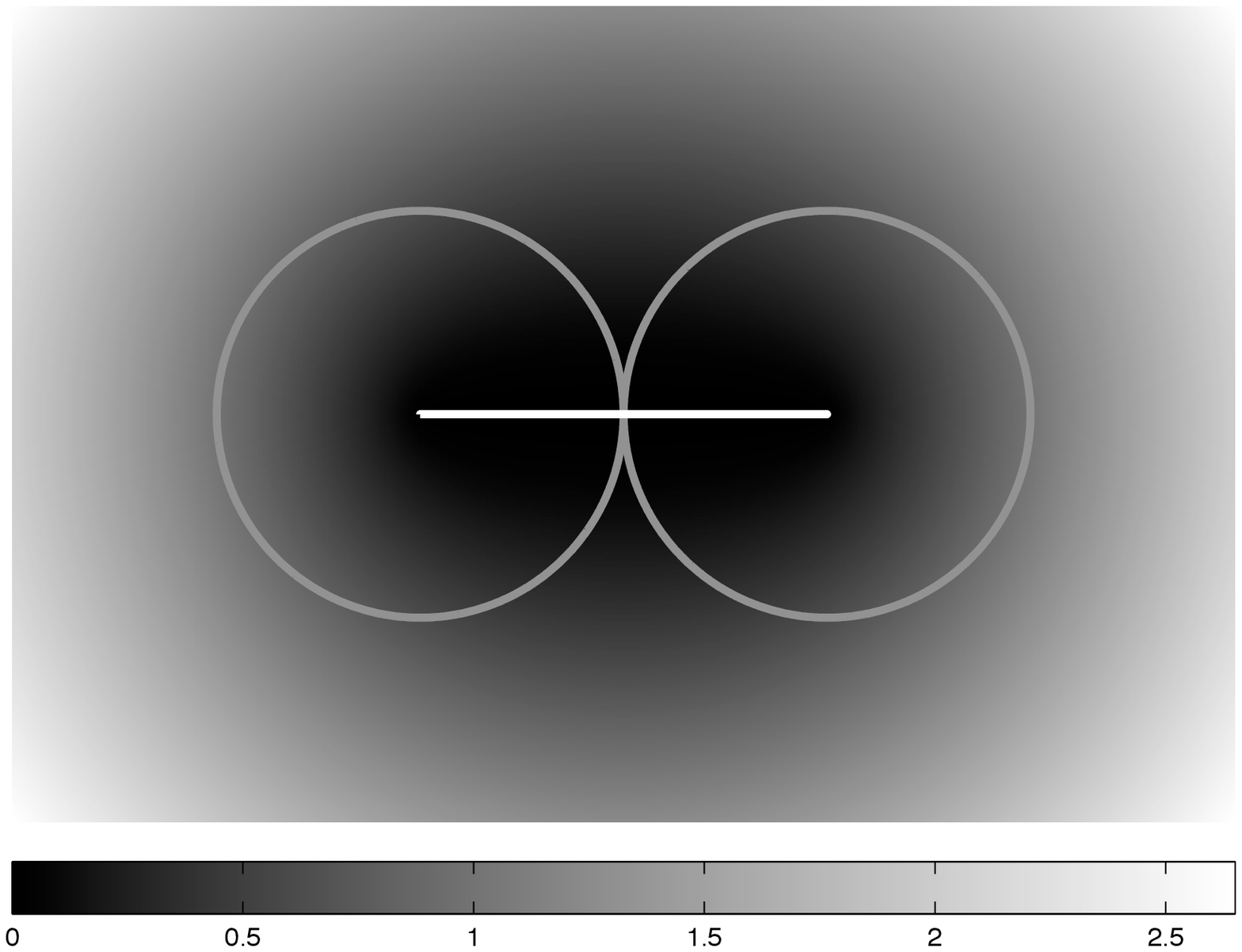} \hspace{3mm} 
    \includegraphics [scale=0.24]{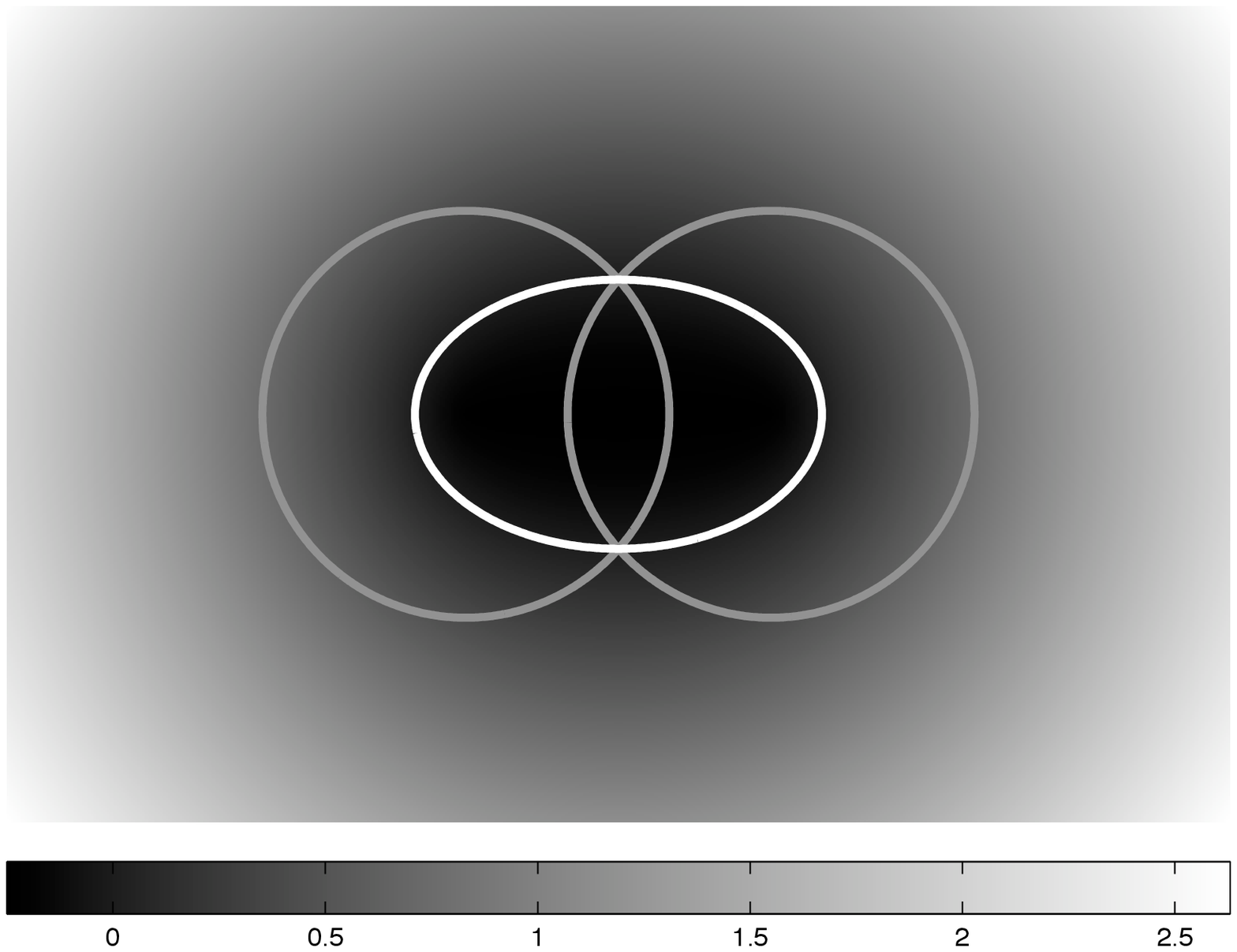} }}
\caption{The expected boundaries (white) for random discs with $p=0.5, r=1$ and $|a| = 3, 2, 1.5$ (from left to right), superimposed on the gray-scale image of the expected ODF. Realized boundaries are shown in gray.} \label{fig:fBallsEll}
\end{figure}
\end{ex}

\begin{ex}[Random half-plane] \label{ex:hplane} Let $\Theta$ be a random variable and consider a random half-plane in $\RofN{2}$ given by
$\boldA=\lbrace x \colon x_1\leq \Theta\rbrace$ with $\partial \boldA=\lbrace x: x_1=\Theta\rbrace.$ The ODF of $\boldA$ is $b_{\boldA}(x)=x_1-\Theta$ and the expectation has a form $ \matexp{b_{\boldA}(x; \Theta)}=x_1-\matexp \Theta.$ From here, the expected set is a half-plane intersecting the horizontal axis at point $\matexp{\Theta}$ and  $\matexp{\partial \boldA}= \lbrace x: x_1=\matexp{\Theta} \rbrace.$ In this example, $\es{\boldA} = \ev{\boldA} = \matexp{\boldA}$. 
\end{ex}

\begin{ex}[Random upper half-plane] \label{example:tplane}
Consider a random upper half-plane in $\RofN{2}$ given by $\boldA=\lbrace x \colon x_2 \geq  x_1 \tan{\Theta}\rbrace$ with $\partial \boldA=\lbrace x \colon
x_2=x_1\tan{\Theta}\rbrace. $ The ODF of $\boldA$ is $b_{\boldA}(x) = |x| \sin{(\Theta-\omega)}$, where $\omega=\arcsin{(x_{2}|x|^{-1})}$.
For $\Theta\sim\unif{a}{b}$ the expected ODF is $\matexp{b_{\boldA(\Theta)}(x)}=2(b-a)^{-1}\sin{((b-a)/2)}|x|\sin{((a+b)/2-\omega)}.$ From here, the expected set is a homothecy with coefficient $\alpha = (b-a)^{-1}(2\sin{(b-a)/2)})^{-1}$ of an upper half-plane with boundary angle $(a+b)/2$. Hence, $\matexp{\boldA}$ is an upper half-plane with boundary given by the line that passes through the origin and makes angle $(a+b)/2$ with the positive horizontal axis. Again, we have $\es{\boldA} = \ev{\boldA} = \matexp{\boldA}$. 
\end{ex}

\section{Separable Oriented Distance Functions}
\label{sec:separable}

Consider a parametric closed set, whose geometry depends on the parameter $\theta\in\RofN{p}$. We call $\theta$ the generating parameter and write $A(\theta)$ for a closed set generated by $\theta$. For example, a disc with a center $\gamma$ and a radius $\rho$ has generating parameters $(\gamma, \rho)$. We now extend the notion of parametric sets to random closed sets. Specifically, let $\Theta$ be a $p$-dimensional ($p\geq 1)$ random variable defined on the probability space $(\Omega, \calP, \bP)$ with $\matexpb{|\Theta|}<\infty$. Consider a random closed set given by a mapping from $\Omega$ to $\calF$, so that the geometry of the set depends on the realized value of $\Theta$. Similarly, we call $\Theta$ the generating parameter of a r.c.s. $\boldA(\Theta)$. For example, a disc with a random center $C$ and a random radius $R$ has generating parameters $\Theta = (C, R)$, whilst a disc with a fixed center and a random radius $R$ has a generating parameter $\Theta = R$. The ODF $b_{\boldA}(x)$ of $\boldA(\Theta)$ is a function of random variable $\Theta$ and so,
is a random quantity. We again assume that $b_{\boldA}(x)$ is integrable for some $x_{0} \in \calD$ and hence is well-defined.

\begin{defn}[Separable ODF] \label{def:sdistSep}
The ODF of set $A = A(\theta)$ with parameter $\theta$ is separable, if it has a form $b_{A}(x; \theta)= h^{T}(x)g(\theta), $
where $h(x) = (h_{1}(x), \ldots, h_{K}(x))^{T}$ and $g(\theta) = (g_{1}(\Theta), \ldots, g_{K}(\theta))^{T}$ for some functions $h_{i}(x), g_{i}(\theta), \; i=1, \ldots, K$.
\end{defn}

\begin{defn}[Separable r.c.s]\label{def:sdistSep}
An r.c.s. $\boldA$ is separable, if there exists a random variable $\Theta$ such that $b_{\boldA}(x; \Theta)$ is separable a.s.. That is
$b_{\boldA}(x; \Theta)= h^{T}(x)g(\Theta) \;\; \textrm{a.s.,}$
with $h(x) = (h_{1}(x), \ldots, h_{K}(x))^{T}$ and $g(\theta) = (g_{1}(\theta), \ldots, g_{K}(\theta))^{T}$ for some functions $h_{i}(x), g_{i}(\theta), \; i=1, \ldots, K$.
\end{defn}

The following theorem establishes the connection between the expectation of a set and its generating parameter.
\begin{thm}\label{proposition2}
Let $\Theta$ be a generating parameter of a separable r.c.s. $\boldA(\Theta)$ and assume that $\matexp{|g_{i}(\Theta)|}<\infty$ for every $i=1, \ldots, K$. Then,
\begin{enumerate}
    \item \label{proposition2:item1} If $g(\theta)$ is convex in every argument, then $\matexp{\boldA(\Theta)} \subset A(\matexp{\Theta})$.
    \item \label{proposition2:item2} If $g(\theta)$ is an affine function of $\theta$ in every component $g_{i}(\theta)$ a.s., then $\matexp{\boldA(\Theta)} = A(\matexp{\Theta})$.
\end{enumerate}
\end{thm}
\begin{proof}
(\ref{proposition2:item1}) It follows from Jensen's inequality that $\matexp{b_{\boldA}(x; \Theta)} = h^{T}(x)\matexp{g(\Theta)} \geq
h(x)^{T} g(\matexp{\Theta}).$ The righthand side of the equation is the ODF of a set with parameter
$\matexp{\Theta}$. For any $x\in \matexp{\boldA(\Theta)}$, we have
$h^{T}(x) g(\matexp{\Theta})\leq \matexp{b_{\boldA}(x; \Theta)}\leq 0$ and
$x\in \cl{A(\matexp{\Theta})}=A(\matexp{\Theta}).$ Thus, $\matexp{\boldA(\Theta)} \subset
A(\matexp{\Theta}).$

(\ref{proposition2:item2}) From Jensen's inequality, $\matexp{b_{\boldA}(x; \Theta)}=h^{T}(x) g(\matexp{\Theta})$ iff the function
$g_{i}(\theta)$ is affine in $\theta$ for every $i=1, \ldots, K$. The expression on the right-hand site is the ODF of a set with parameter $\matexp{\Theta}$, hence the statement \eqref{proposition2:item2} of the theorem follows. 

Note that although the linearity of $g(\theta)$ is a necessary and sufficient condition for the form of $\matexp{b_{\boldA}(x; \Theta)}$, it is only a sufficient condition for the separability of $\matexp{\boldA(\Theta)}$. This is because the expected set depends on $\matexp{b_{\boldA}(x; \Theta)}$ up to a constant, so that for fixed $c>0, \matexp{b_{\boldA}(x; \Theta)}$ and $c\matexp{b_{\boldA}(x; \Theta)}$ define the same expected set. 
\end{proof}

\begin{ex}\label{ex:separable}
In Example \ref{ex:ball2}, the ODF of a closed ball with random radius $b_{\boldA}(x) = |x|-\Theta$ is separable with $h(x)=(|x|, -1)^{T}$ and $g(\theta)=(1, \theta)^{T}$, and the expected set is a closed ball with radius $\matexp{\Theta}.$

In Example \ref{example:discreteCircle}, the ODF of flashing discs is not separable, and the geometry of the expected set differs from that of the observed sets. However, note that both the random set and its expectation can be described as Cartesian ovals.

In Example \ref{ex:hplane}, the ODF of a half-plane is separable with $h(x)=(x_{1}, -1)^{T}$ and $g(\theta) = (1, \theta)^{T}$. Hence, the expected set is a half-plane with parameter $\matexp{\Theta}$.

In Example \ref{example:tplane}, the ODF of an upper half-plane can be written as $b_{\boldA}(x) = x_{1}\sin\Theta - x_{2}\cos\Theta$, and so is separable with $h(x)= (x_{1}, -x_{2})^{T}$ and $g(\theta) = (\sin\theta, \cos\theta)^{T}$. The expected set is an upper half-plane whose boundary makes angle $\matexp{\Theta}$ with the positive horizontal axis. This example illustrates the second statement of Theorem \ref{proposition2}, in that $\matexp{\boldA(\Theta)} = A(\matexpb{\Theta})$, even for nonlinear $g$.
\end{ex} 

\begin{rem} Note that if $g$ is invertible in every component, we can reparametrize the expected ODF as $\matexp{b_{\boldA}(x; \Theta)}= h^{T}(x)g(\eta)$, where $\eta_{i} =g_{i}^{-1}(\matexpb{g_{i}(\Theta)}), i=1, \dots, K$. Hence, the expected set and the expected boundary have the same geometric form in terms of parameters $\eta_{i}$. The relation between $\eta_{i}$ and the moments of $\Theta$ is established on a case-by-case basis. 
In particular, for the ODFs of the form $b_{\boldA}(x) = h(x) + g(\Theta)$ for some function $h(x)\colon \RofN{d}\mapsto \mathbb{R}$ and an invertible function $g(\theta)\colon \RofN{p}\mapsto \mathbb{R}$, we obtain that $\matexp{b_{\boldA}(x; \Theta)}=h(x)+ \matexp{g(\Theta)}=h(x)+ g(g^{-1}(\matexp{g(\Theta)}))$ and $\matexp{\boldA(\Theta)} = A(g^{-1}\matexp{g(\Theta)}).$ 
\end{rem}

\begin{rem}
In digital processing, data are stored as an array of pixels and so the ODF of an object in the image is discretised. This discretization implies that for realizations of an r.c.s. $\boldA$ at times $t=1, 2, \ldots$, the ODFs of $\boldA_{t}$ can be viewed as a random sample of arrays in the time domain. For a separable ODF, the space-time covariance matrix has a form
\begin{displaymath}
\cov(b_{\boldA_{t}}(x), b_{\boldA_{\tau}}(y)) = h^{T}(x)\cov(g(\Theta_{t}), g(\Theta_{\tau}))h(y), 
\end{displaymath}
for lattice vertices $x, y$ and time points $t, \tau$. This expression shows that the covariance of a separable ODF is a product of a deterministic spatial component and the covariance of the stochastic time component. In particular, for an ODF of the form $b_{\boldA}(x) = h(x)+g(\Theta)$, the space-time covariance is given by $\cov(b_{\boldA_{t}}(x), b_{\boldA_{\tau}}(y)) = \cov(g(\Theta_{t}), g(\Theta_{\tau}))$, and hence the dependence between the sets can be inferred from their generating parameters. 
\end{rem}
\section{Further Developments and Applications}
\label{sec:applications}

\subsection{The Sample Mean Set}
\label{sec:sampleMean}

Let $\boldA_1, \ldots, \boldA_m$ be the observed realizations of an r.c.s. $\boldA$. For every $i$, we assume that $\boldA_i\subset \cl{\calD}$
and $\partial \boldA_i\neq \emptyset$ and denote by $b_i(x)$ the ODF of $\boldA_i$ at point $x$. Let  $\bar{b}_m(x)=m^{-1}\sum_{i=1}^m b_i(x)$ be the sample mean ODF at point $x$.

\begin{prop}\label{prop:sampleMeanSdist}
  The sample mean ODF $\bar{b}_m(x)$ is uniformly Lipshitz with constant one.
\end{prop}
\begin{proof}
The proof follows immediately from the Lipshitz continuity of $b_{i}(x)$.
\end{proof}

\begin{defn}[Sample mean set]\label{def:sdist_empSet}
The sample mean set $\bar{\boldA}_m$ and the sample mean boundary $\partial \bar{\boldA}_m$ are given by the zero-level set and by the zero-level isocontour of the empirical ODF, respectively, i.e.
\begin{equation}\label{eq:sdist_empset}
  \bar{\boldA}_m=\lbrace x: \bar{b}_m(x)\leq 0\rbrace, \quad \partial \bar{\boldA}_m=\lbrace x: \bar{b}_m(x) =0\rbrace.
\end{equation}
\end{defn}
\noindent The sample mean set defined in \eqref{eq:sdist_empset} is a random closed set, since $\lbrace \bar{\boldA}_{m}\cap K \neq \emptyset \rbrace = \lbrace \inf_{x\in K} \bar b_{m}(x) \leq 0 \rbrace$ is measurable. 

\begin{figure}[h]
{\footnotesize \psfrag{samplesize}[t][c]{Sample size, $m$}
\psfrag{error}[b][t]{} \centerline{
    \includegraphics [scale=0.38]{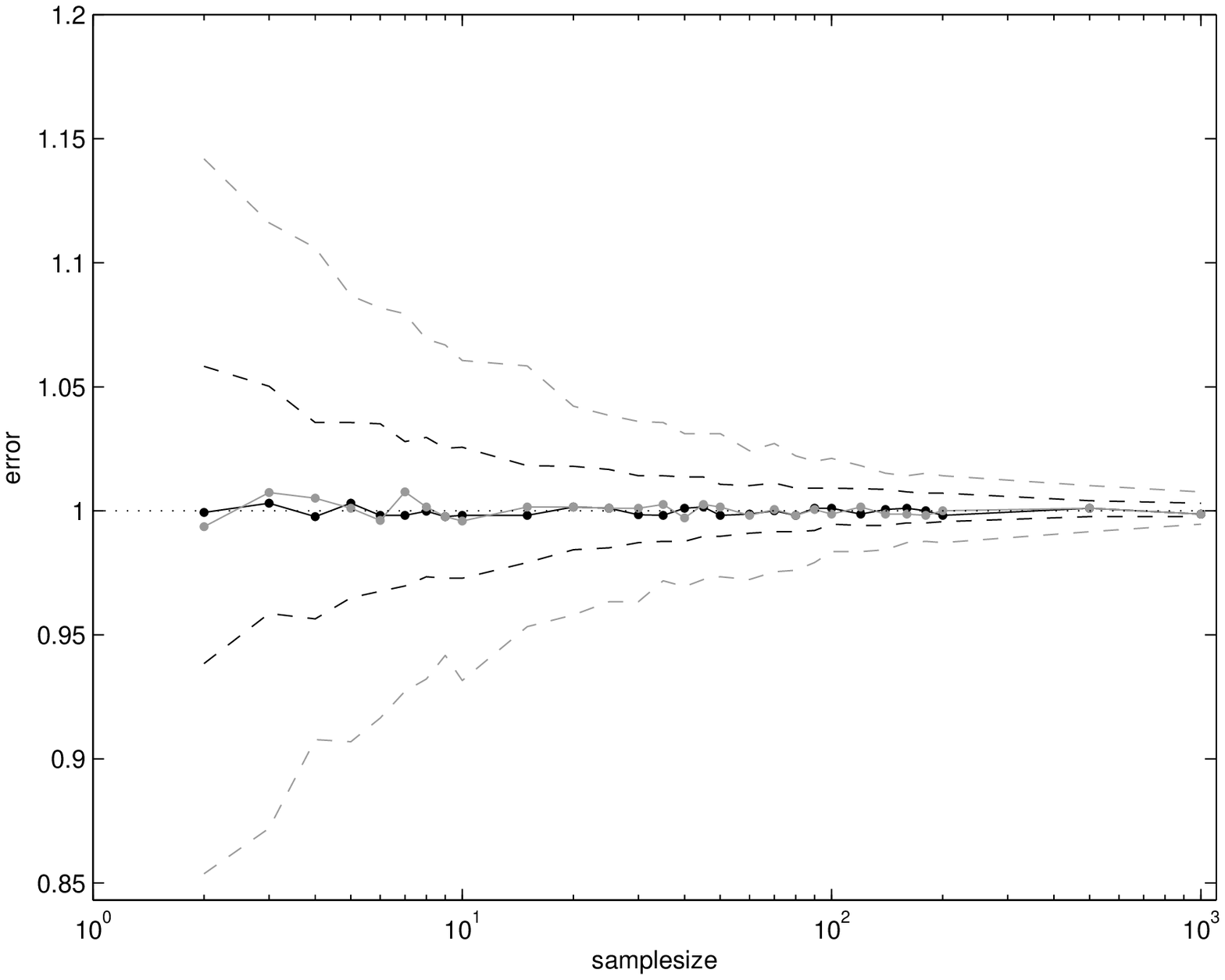}
        \includegraphics [scale=0.38]{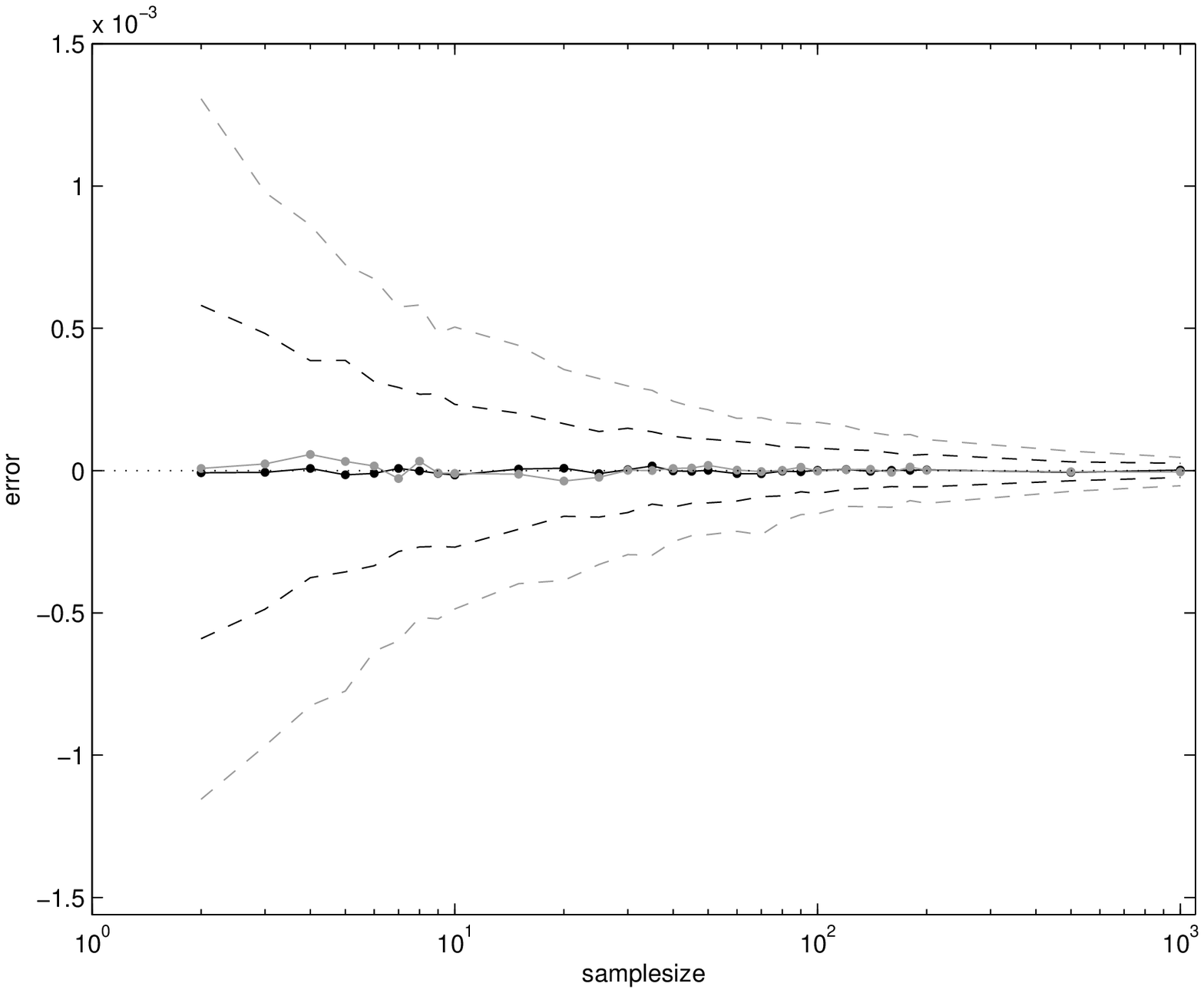}
   }}\caption{(Left) The ratio of the radii of $\bar\boldA_{m}$ to that of the expected set $\matexp{\boldA}$ in Example \ref{ex:ballSM} for
$\Theta \sim \unif{0.8}{1.2}$ (black) and $\Theta \sim \unif{0.5}{1.5}$ (gray). (Right) The difference between the boundary angles of $\bar\boldA_{m}$ and $\matexp{\boldA}$ in Example \ref{ex:planeSM} with $\Theta \sim \unif{\pi/8}{3\pi/8}$ (black) and $\Theta \sim \unif{0}{\pi/2}$ (gray). The median (solid) and the 25-th and 75-th percentile (dashed) values are based on 1000 simulations for each $m$.} \label{fig:sample}
\end{figure}

\begin{ex}[Disc with random radius] \label{ex:ballSM} Recall that for a closed ball $\boldA$ with a random radius $\Theta$, $\matexp{\boldA(\Theta)} = A(\matexp{\Theta})$ and $\matexp{\partial \boldA} = \partial A(\matexp{\Theta})$. For a sample of closed balls $\boldA_1, \ldots, \boldA_m$ with random radii $\Theta_1, \ldots, \Theta_m$, the emprirical ODF is $\bar{b}_m(x)=|x|-\bar{\Theta}_m$ and hence $\bar{\boldA}_m = A(\bar{\Theta}_m)$ and $\partial\bar{\boldA}_m = \partial A(\bar{\Theta}_m)$.

For an independent random sample of radii values $\Theta_{1}, \ldots, \Theta_{m}$, we construct $m$ random discs in $\RofN{2}$ centered at the origin. To compare the empirical boundary estimator with the theoretical average, we compute the ratio of the radii of $\matexp{\boldA}$ to $A(\bar{\Theta}_m)$; see Figure \ref{fig:sample}, left. The radii values were sampled from the uniform distribution $\unif{0.8}{1.2}$ (black) and $\unif{0.5}{1.5}$ (gray). The median values of the ratios (solid curves) and the 25-th and 75-th percentiles (dashed curves) are based on 1000 samples for each sample size $m$. The figure shows that the accuracy of the empirical estimate is higher for bigger sample sizes and for smaller variance values of the generating parameter.
\end{ex}

\begin{ex}[Random upper half-plane] \label{ex:planeSM} Consider an r.c.s. in Example \ref{example:tplane} given by $\boldA(\Theta) = \lbrace x\in \RofN{2} \colon x_{2} \geq x_{1}\tan \Theta \rbrace$, where $\Theta$ is the angle between the boundary of the plane and the positive horizontal axis. Recall that for $\Theta\sim\unif{a}{b}, \; \matexp{\boldA(\Theta)} = A(\matexp{\Theta})$ and $\matexp{\partial \boldA} = \partial A(\matexp{\Theta})$.

For an independent random sample of the angle values $\Theta_{1}, \ldots, \Theta_{m}$, we construct $m$ upper half-planes. The empirical mean boundary is estimated by the zero-level isocontour of the empirical mean ODF. Figure \ref{fig:sample} (right) shows the difference between the boundary angles for the empirical mean set and that of the expected set for $\Theta \sim \unif{\pi/8}{3\pi/8}$ (black) and $\Theta \sim\unif{0}{\pi/2}$ (gray). The median values of the ratios (solid curves) and the 25-th and 75-th percentiles (dashed curves) are based on 1000 replicates for each value of $m$. As in the previous case, the boundary estimate is more accurate for bigger sample sizes and for smaller variance values of $\Theta$. 
\end{ex}

The empirical mean set and the empirical boundary provide a recipe for practical construction of the set and boundary estimators. Examples \ref{ex:ballSM} and \ref{ex:planeSM} suggest the consistency of the boundary estimator, based on the empirical ODF. Note that both the expected set $\matexp{\boldA}$ and the empirical mean set $\bar\boldA_{m}$ are described as level sets of the continuous function $b_{\boldA}$ and hence, under certain conditions, the consistency of $\bar{\boldA}_{m}$ as an estimator of $\matexp{\boldA}$ can be inferred from \citep[Theorem 2]{Molchanov1998}. The result of \citep{Molchanov1998}, however, does not apply to the boundary estimator $\partial \bar{\boldA}_{m}$. In our upcoming paper, we study the consistency of the boundary estimator in greater detail \citet{JankowskiStanberryXX}.

\subsection{Loss Functions}

Assume that the observed data $\boldA \in \calF$ follows a distribution $F_{\parA}$ with mean parameter $\parA = \matexp{\boldA}, \; \parA \in\calF$. We call $\parA$ the parameter set. Here, data $\boldA$ represent random closed sets with nonempty boundary. In practice, it is more common to acquire images rather than observe sets per se, so the observed $\boldA$ can be seen as a result of the first-step analysis.  The goal is to estimate the parameter set $\parA$ given the observed data $\boldA$.

If set $\parA$ is parametric with parameter $\theta$, we can construct an estimator of $\parA$ using conventional loss functions for $\theta$. For example, if $\boldA$ is a collection of deformed balls, the distribution $F$ can be parametrized by the center and radius of a ball. Depending on the problem, we might be interested in testing the hypothesis about either the location or the size of the ball, or both. Then the optimal estimator $\widehat\theta$ is determined using standard loss functions, e.g. the indicator loss $\ell(\theta, a)=\ind{a=\theta}$, the absolute error loss $\ell(\theta, a) = |a-\theta|$ or the squared error loss $\ell(\theta, a) =(a-\theta)^{2}$. However, often, the sets of interest cannot be parametrized and at best can be described as ``blobs'', in which case conventional loss functions are not applicable. Here, we relay loss functions for Bayesian inference about sets introduced in \citet{StanberryBesagXX} in the frequentist framework. 

To begin with, we consider a loss function based on representation of a set $A$ given by its characteristic function $\charx{A}{x} = 1_{A}(x)$. The space of characteristic functions is complete with metric
\begin{equation}\label{metric:chi2}
d(A, B)=||\charomega{A}-\charomega{B}||_{L^q} =\left(\int_{\calD} |\charx{A}{x}-\charx{B}{x}|^q \lambda(\dd x)\right)^{1/q}.
\end{equation}
\noindent Here and later, the integration domain $\calD$ is either the entire image or its subset.
Let $\hatchi{}$ be an estimator of $\charomega{A}$. Given the metric $d$, we define an $\Lp{1}$ loss function as
\begin{equation}\label{loss:chi1}
\ell_{1}(\charomega{\parA},\hatchi{})=||\charomega{\parA}-\hatchi{}{}||_{L^1} =\int_{\calD} |\charx{\parA}{x}-\hatchi{}{}| \lambda(\dd x).
\end{equation}
Because the estimator $\hatchi{}$ is the characteristic function that minimizes the expected loss, the estimator $\widehat \parA$ of $\parA$ is immediately available. Strictly speaking, $d(A,B)$ distinguishes between sets up to a set of Lebesgue measure zero, so we assume that the parameter set has no punctures and hence $\widehat \parA$ is uniquely determined. Note that the equation \eqref{loss:chi1} is the measure of the symmetric difference, since 
\begin{equation}\label{loss:chiL1}
\ell_{1}(\charomega{\parA}, \hatchi{}) = \lambda(\parA \Delta \widehat\parA) = \lambda((\parA \setminus \widehat\parA) \cup (\widehat\parA \setminus \parA))
\end{equation}
and, hence, $\ell_{1}(\charomega{\parA}, \hatchi{})$ is akin to the absolute error loss in conventional settings. The set that minimizes $\matexp{\lambda(\parA \Delta \widehat \parA)}$ is the Vorob'ev median, i.e. is the median level set of the coverage function \eqref{eq:covFn}; \citep[see][p.~178]{Molchanov2005}. The Vorob'ev median is a set-analog of the ordinary median in the sense that, as the latter, the former minimizes the absolute error loss.

Moreover, the Vorob'ev median is optimal with respect to a much broader class of $\Lp{q}$ loss functions,
\begin{equation}\label{loss:chi}
\ell_{q}(\charomega{\parA},\hatchi{})=||\charomega{\parA}-\hatchi{}{}||_{L^q}, \quad 1\leq q <\infty.
\end{equation}
This is because $\charomega{\parA}$ takes values in $\lbrace 0, 1 \rbrace$ and the induced $\Lp{q}$ topologies \eqref{loss:chi} are all equivalent for $1\leq q < \infty$; \citep[see][Theorem~2.2]{DelfourZolesio2001}. 

Instead of optimizing the expected $\Lp{q}$ loss over the set of characteristic functions, consider minimizing \eqref{loss:chi} over a larger set of functions with values in $[0,1]$. For the squared $\Lp{2}$ loss
\begin{equation}\label{loss:chi2}
\ell_{2}(\charomega{\parA}, \hatchi{})=||\charomega{\parA}-\hatchi{}||^2_{\Lp{2}},
\end{equation}
the pointwise estimator $\hatchix{}{x}\in [0, 1]$ is the coverage function \eqref{eq:covFn}. The estimator $\hatchix{}{x}$ is a global minimum in $\Lp{2}$ since 
$\matexp{||\charomega{\parA}-\matexp{\charomega{\parA}}||^{2}_{\Lp{2}}} \leq \matexp{||\charomega{\parA}-\hatchi{}||^{2}_{\Lp{2}}}.$
Note that $\hatchi{}$ is essentially a gray-scale image with intensities reflecting the probability of a pixel being in $\parA$. 

The coverage function, however, is not necessarily a characteristic function, unless $\boldA$ is deterministic a.s.. Consider a set-valued estimator given by an excursion set of $\hatchix{}{x}$ so that determining the optimal estimator $\widehat{\parA}$ now reduces to choosing an appropriate threshold $u$ for the excursion set \eqref{eq:exSets}. The Vorob'ev criterion \eqref{def:vor} gives an estimator $\widehat \parA$, which is optimal with respect to Lebesgue measure in the sense that $ \matexp{\lambda(\parA \Delta \widehat\parA)}\leq \matexp{\lambda(\parA \Delta \widehat\parA)}$ for all measurable sets with measure $\matexp{\lambda(\boldA)}$ \citep[see][Theorem 2.3, pg.~177]{Molchanov2005}. From \eqref{loss:chiL1} and the equivalence of the $\Lp{q}$-topologies on the space of characteristic functions, it follows that the Vorob'ev expectation minimizes the expected loss $\eqref{loss:chi}$ for any $q\in [1, \infty)$ over all of the sets with Lebesgue measure $\matexp{\lambda(\boldA)}$. 

In general, given the representative functions $f_{A}$ and $f_{B}$ of sets $A$ and $B$, the discrepancy between $A$ and $B$ can be taken in terms of some (pseudo-) metric
\begin{equation}\label{dist:Badd}
\calm(A, B) = \calmw(f_{A}, f_{B}), 
\end{equation}
defined on a compact window $\calW\subseteq\calD$. Further, we suppress the subscript $\calW$ whenever possible. 
Given the metric $\calm(A, B) = ||f_{A} -f_{B}||_{\Lp{2}},$ the global estimator based on the squared $\Lp{2}$ loss
\begin{equation}\label{loss:Badd}
\ell_{2}(f_{\parA}, \widehat f) = ||f_{\parA} - \widehat f||_{\Lp{2}}^{2}
\end{equation}
is the average function $\matexp{f_{\boldA}}$. Hence, the distance-average expectation \eqref{def:daExp} is an optimal estimator of $\parA$ with respect to \eqref{loss:Badd}, among all of the level sets of $\matexp{f_{\boldA}}$.

Note that although the distance-average expectation arises naturally for the $\Lp{2}$ loss, it was originally defined for a generic (pseudo-) metric $\calm(\cdot, \cdot)$ \citep{BaddeleyMolchanov1998}. Thus, in principle, we can construct an estimator as the level set of the expected representation for arbitrary choice of $\calm$. For that, we must specify the representative function $f$, select an appropriate (pseudo-)metric $\calm$, and choose the loss function $\ell(f_{\parA}, \widehat f)$ based on $\calm$. The threshold $u$ in \eqref{def:daExp} then determines the optimal estimator with respect to $\ell(f_{\parA}, \widehat f)$ among all of the level sets of $\matexp{f_{\boldA}(x)}$. 

We now consider a representation of a set given by its ODF. Recall that the space of ODFs of sets with non-empty boundary is complete with metric $d(A, B)=||b_B-b_A||_{\Lp{q}}.$ For the loss function 
\begin{equation}\label{loss:odf}
\ell_{2}(b_{\parA}, \widehat{b})=||b_{\parA}-\widehat{b}_{\;}||^2_{\Lp{2}} = \int_{\calD} |b_{\parA}(x)-\widehat{b}(x)|^2 \lambda(\dd x),
\end{equation}
let $\widehat b$ be the estimator of $b_{\parA}$ that minimizes the expectation of \eqref{loss:odf}. Because the ODF gives a unique (up to the boundary) representation of a set, the estimator $\widehat \parA$ is given by the zero-level set of $\widehat b$. However, optimizing $\matexp{\ell_{2}(b_{\parA}, \widehat{b})}$ over the set of ODFs is nontrivial, so instead we consider unrestricted minimization over the set of functions defined on $\calD$. In this case, the pointwise estimator is the expected ODF, which is also a global minimum.  
Because the estimator $\matexp{b_{\boldA}(x)}$ is not necessarily an ODF itself, it does not uniquely determine the set. Hence, we take the set-valued estimator $\widehat{\parA}$ of $\parA$ to be the zero-level set of $\matexp{b_{\boldA}(x)}$, so that $\widehat \parA$ is the expected set as defined in \eqref{def:sdistMeanSet}. The estimator of the boundary $\partial \parA$ is given by the zero-level isocontour of $\matexp{b_{\boldA}(x)}$.


\subsection{Image Averaging}

In this section, we consider the example of image averaging originally discussed in \citet{BaddeleyMolchanov1998}. In image averaging, which relates to Bayesian image classification and reconstruction, the goal is is to determine an average object or a typical shape from the collection of images of the same scene or objects of the same type. 

To begin with, we quickly outline the image sampling procedure and refer for more details to \citet{BaddeleyMolchanov1998}. Figure \ref{fig:bmResults} (left) shows the true binary image, $I$, of the scanned and thresholded newspaper fragment. After adding Gaussian noise, 15 independent realizations from the posterior were obtained using a Gibbs sampler with true noise parameters in the likelihood and an Ising prior. In the current notations, $\calD = \calW = I$, $\parA$ is the true newspaper text and $\boldA_{1}, \ldots, \boldA_{15}$ are the 15 independent reconstruction of $\parA$. The data for this example was downloaded from \citet{BaddeleyWeb}.

\begin{figure}[!hbtp]
\centerline{
    \fbox{\includegraphics [scale=0.39]{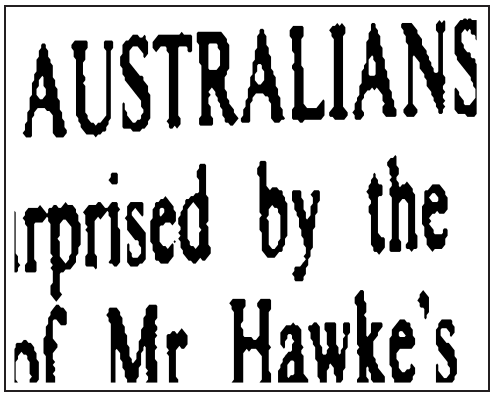}}
    }
\caption{The true binary image of the newspaper fragment.}  \label{fig:bmResults}
\end{figure}

Figure \ref{fig:bmResults} shows the ODF estimator $\widehat \parA_{ODF}$ (left) and the distance-average estimator $\widehat \parA_{DA}$ (right) with $f_{A}(\cdot) = b_{A}(\cdot)$ and $\calm(\cdot, \cdot) = \Lp{2}(\cdot, \cdot)$; see also Figure 6 in \citet{BaddeleyMolchanov1998}. The estimator $\widehat \parA_{ODF}$ appears less noisy and more accurate as compared to $\widehat \parA_{DA}$, where the letters look overinflated. 

\begin{figure}[!hbtp]
\centerline{
\fbox{\includegraphics [scale=0.39]{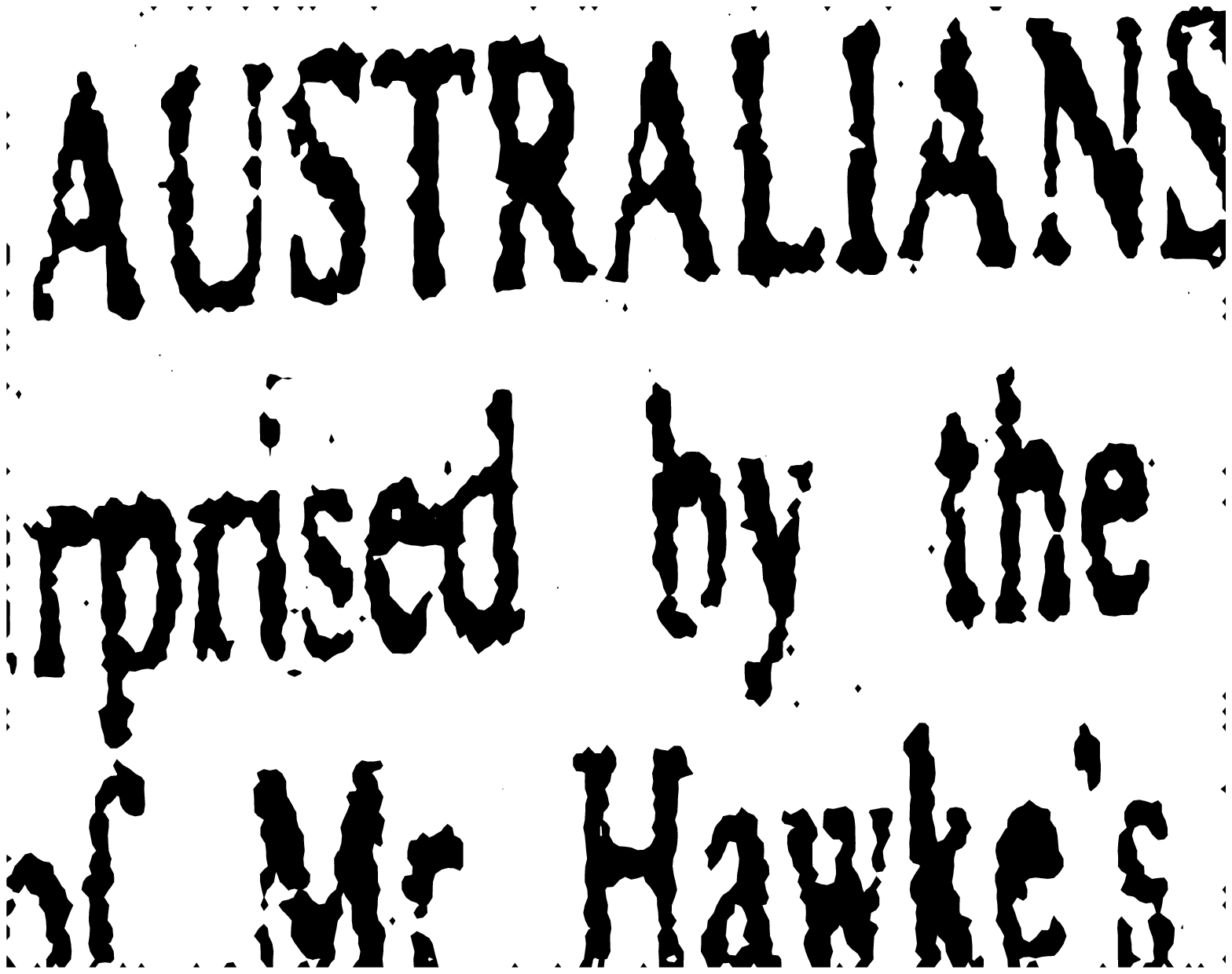}}
   \fbox{\includegraphics [scale=0.39]{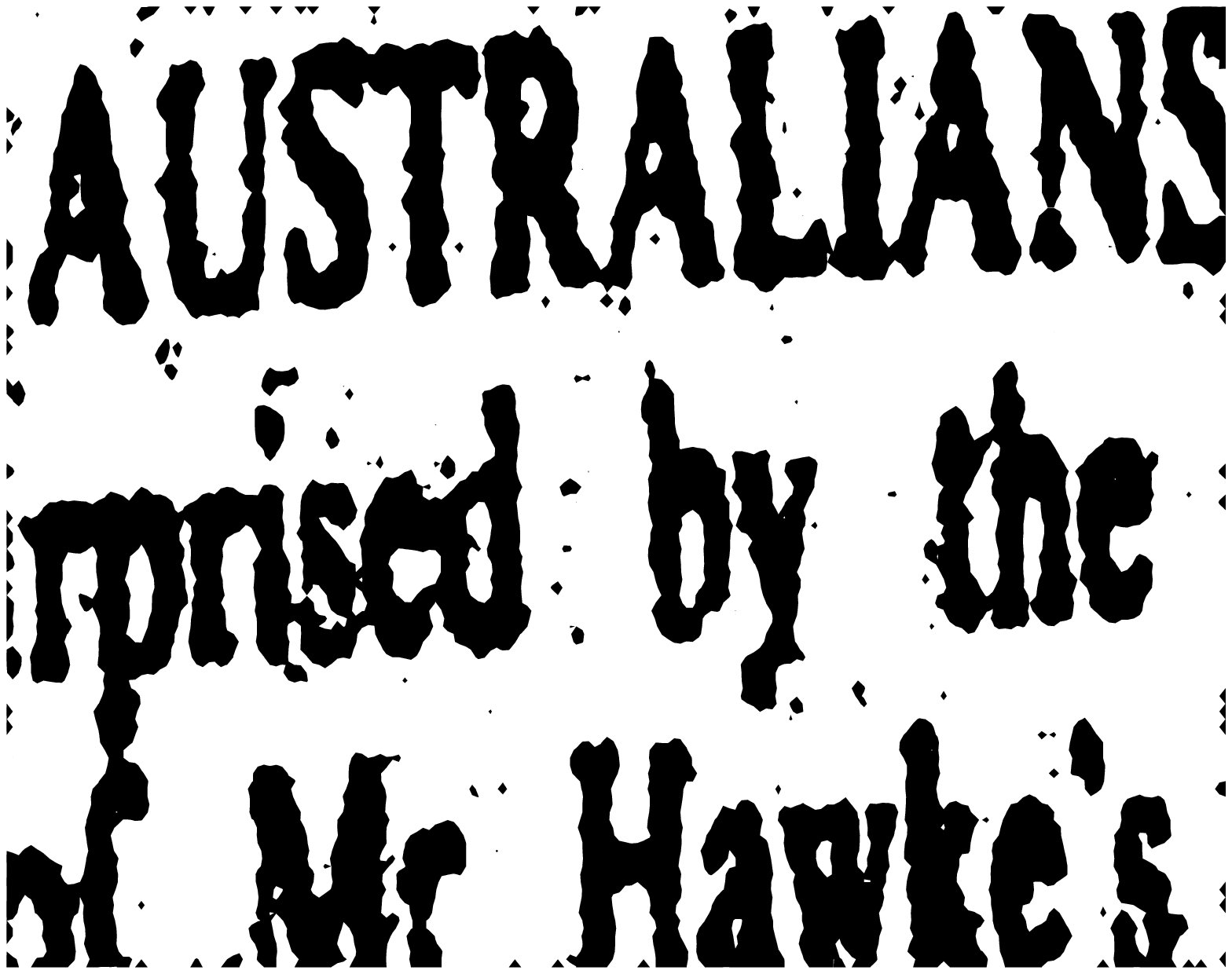}}
    }
\caption{The ODF estimator $\widehat \parA_{ODF}$ (left) and the distance-average estimator $\widehat \parA_{DA}$ (right).} \label{fig:bmAvODF}
\end{figure}

To compare the results of the reconstruction, it is instructive to look at the residual images for $\widehat \parA_{ODF}$ and $\widehat \parA_{DA}$ in Figure \ref{fig:bmResid}, left and right, respectively. Note that black (resp. gray) pixels in the figure correspond to the set $\widehat \parA \setminus \parA$ (resp. $\parA \setminus \widehat \parA$). The residual image for $\widehat \parA_{DA}$ shows a clear spatial pattern, indicating that $\widehat \parA_{DA}$ tends to overestimate the set. In turn, the residuals of $\widehat \parA_{ODF}$ show little spatial clustering. In addition, $\widehat \parA_{ODF}\subset \widehat \parA_{DA}$, that is the distance-average estimator is inclusive of the ODF one, with discrepancies observed along the boundary and in the background. 

\begin{figure}[!hbtp]
\centerline{
    \fbox{\includegraphics[width=0.48\textwidth, height=0.245\textheight]{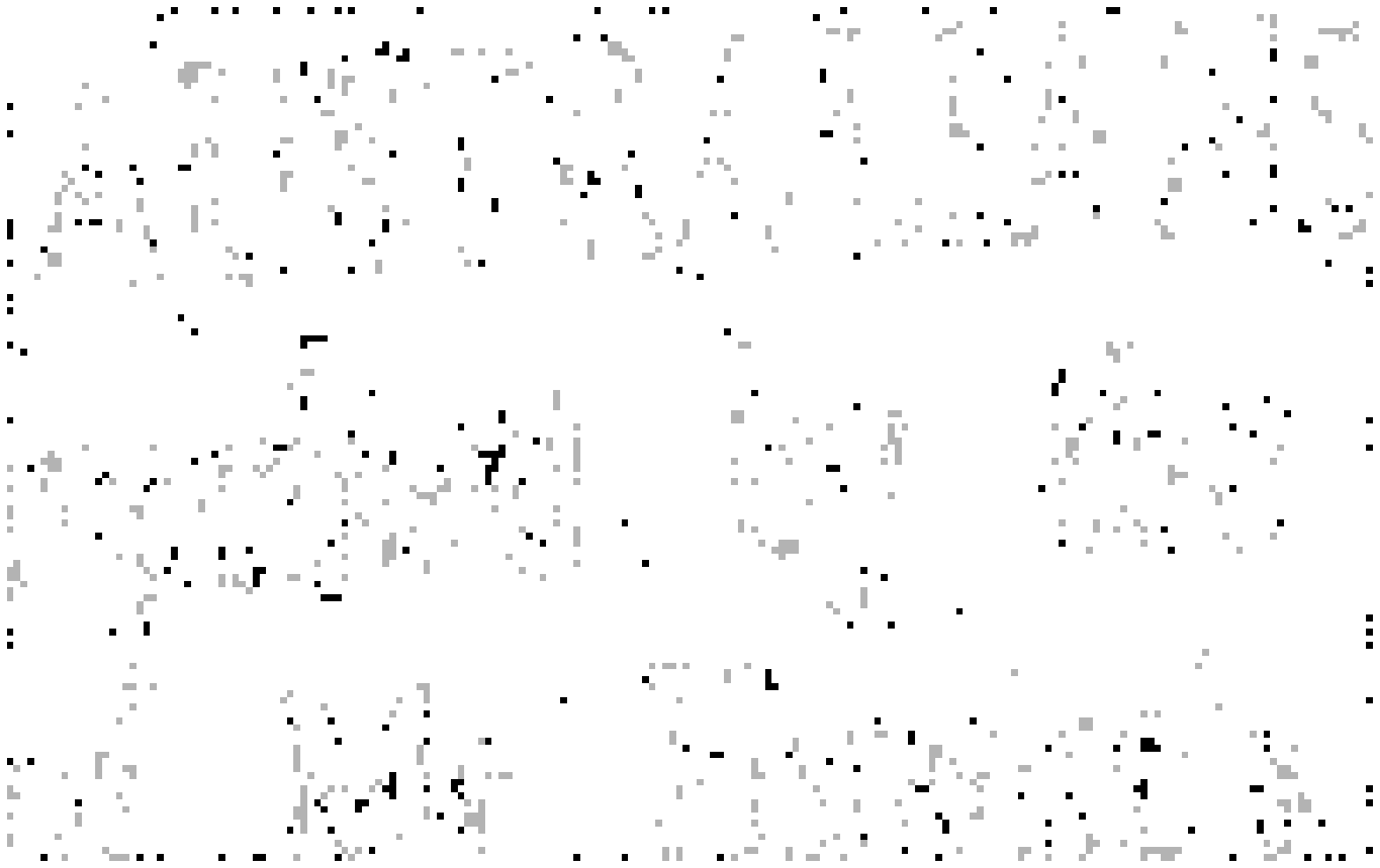}}
    \fbox{\includegraphics[width=0.48\textwidth, height=0.245\textheight]{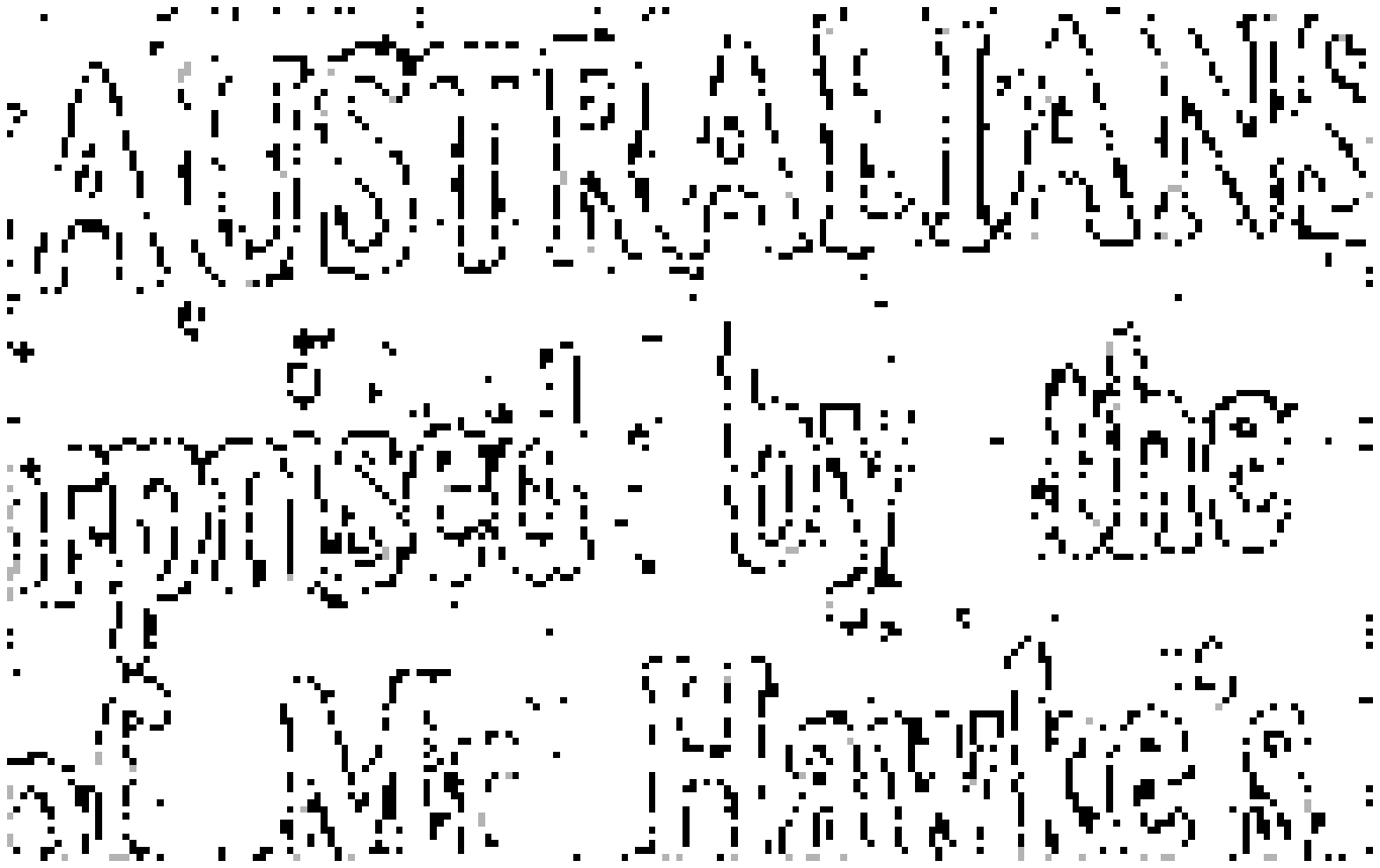}}
    }
\caption{The residual images of the discrepancy between the estimated set and the truth for $\widehat \parA_{ODF}$ (left) and $\widehat \parA_{DA}$ (right). Black (resp. gray) pixels mark set $\widehat \parA \setminus \parA$ (resp. $\parA\setminus \widehat \parA$).}\label{fig:bmResid}
\end{figure}

To compare the quality of the estimators, it is common to use the fraction of misclassified pixels. Here, the error is 4.12\% for $\widehat \parA_{ODF}$ and 10.68\% for $\widehat \parA_{DA}$. However, a small misclassification error does not necessarily imply similar looking images \citep{Baddeley1992}, therefore, we compare the two estimators by computing an $\Lp{2}$ distance to the truth
\begin{displaymath}
d(\parA, \widehat \parA) = (\int_{\calD} |b_{\parA}(x) - b_{\widehat \parA}(x)|^{2} \dd x)^{1/2},
\end{displaymath}
which is compatible with the distance-average expectation. Calculating the distances for the two estimators, we obtain $d(\parA, \widehat \parA_{ODF})  = 1.78$ and $d(\parA, \widehat \parA_{DA})=2.66$, which shows that the ODF estimator outperforms the distance-average reconstruction.

Note that the reconstruction results using ODFs do not change if we reverse the color of pixels, as opposed to the distance-average estimator with $f_{A}(\cdot) = d_{A}(\cdot)$, where the swapping affects the quality of the reconstruction \citep[see][Figures 4 and 5]{BaddeleyMolchanov1998}.



\subsection{Boundary Reconstruction in Noisy Images}

In image analysis, it is often required to reconstruct the boundary of an object in a noisy image. Low level methods such as edge detection and intensity thresholding strongly depend on image characteristics and require a careful tuning of parameters to produce satisfactory results \citep{Canny1986, Gonzalez2002}. A number of methods for boundary reconstruction have been proposed in the literature. In active contour models, the boundary is represented by a parametric or a free-form curve and the reconstruction is driven by local image forces, with smoothness controlled via user-imposed regularity constraints \citep{BlakeIsard1998, Brigger2000, ElBaz2005, Flickner1994, Kass1988, Menet1990, Rueckert1997, Tauber2004, Terzopoulos1987, Yuille1992}. The use of local properties of the image implies strong dependence of the boundary estimate on initialization of the modeling curve, a contrast-to-noise ratio of the image and convexity of the object. Various extensions were proposed in the literature to improve the performance of active contours \citep{Coughlan2000, ElBaz2005, Wang1996, Cohen1991, Xu1998}.

Many Bayesian methods for boundary reconstruction use templates which represent local or global geometry of an object \citep{Grenander1976, Tjelmeland1998, Amit1991, Grenander1996, Jain1996, GrenanderKeenan1993, Grenander1996, Jain1998, Chalmond2003, Geman1990}. The stochastic method in \citet{Qian1996} models the boundary by closed polygons. The method was then extended to allow for the dynamic estimation of the number of polygon vertices, thus balancing the complexity and the flexibility of the model \citep{PievatoloGreen1998}. However, the reconstruction is given by a gray-scale image, where intensities reflect the posterior mean state of pixels, so that constructing a boundary estimator requires additional post-processing. 

In \citet{StanberryBesagXX}, the authors proposed a new method to reconstruct a smooth connected boundary of an object in a noisy image using B-spline curves. The method is Bayesian and uses a Markov chain Monte Carlo algorithm to obtain curve samples from the posterior. Using the squared $\Lp{2}$-loss based on ODFs, the posterior estimator is given by the posterior expected boundary as defined in \eqref{def:odfMean}. The simulation study showed that the ODF estimator is more accurate as compared to the distance-average reconstruction and Vorob'ev sets. 

\subsection{Implementation}

In addition to its appealing theoretical properties, the ODF expectation can be efficiently computed, using algorithms for the distance function \citep{Breu1995, Freidman1977, Rosenfeld1966}. The ODF is computed by combining the distance function for the original image and that of the inverted image, i.e. the black and white colors are interchanged. The examples and results reported in this paper were coded in MATLAB (R2008A, The MathWorks), where the distance function to a set is computed using the \texttt{bwdist} command. 

In comparison to the distance-average and Vorob'ev expectations, the ODF average is more efficient, because it requires no optimization since the threshold level is fixed at zero. This difference is particularly prominent in Bayesian reconstruction, where the estimator is based on thousands of samples from the posterior. In addition, the storage requirements for the ODF reconstruction in Bayesian framework are miniscule, as we only need to update the average after every sweep.

\section{Discussion}
\label{sec:discussion}

In this paper, we present new definitions of the expected set and the expected boundary, using the ODF representation of a set. In conventional settings, where the parameter of interest is a scalar or a vector, it is straightforward to construct an estimator of the mean from the observed data. However, statistical inference about sets is nontrivial. When dealing with sets, one has to determine the features that are important to emphasize and develop inference methods based on them. These features may include the location, the size, or the orientation of a set. If a set has an analytic representation, the most natural approach to set inference is via inference about its parameters. This, however, is not applicable for sets with arbitrary geometry. 

A number of existing definitions of the expected set are based on the linearization idea outlined in \citet{Molchanov2005}, where inference is made on the space of functions representing a set. The resulting estimator necessarily depends on the type of linearization and the choice of the representative function.  In this paper, we use an approach similar in spirit to the linearisation idea and define the expectation based on the ODF representation of a set. However, we forgo the optimality criterion and simply choose zero as a threshold. Whilst the choice might appear simplistic, it gives an expectation with attractive theoretical traits, including equivariance, convexity-preservation and inclusion properties. 

The definition is particularly appealing for problems in image analysis. In particular, its denoising property implies that random specks will be averaged out unless observed with probability one. In comparison, the distance-average expectation is not empty by construction, and therefore identifies false positives in any noisy image. Furthermore, the distance-average reconstruction strongly depends on the choice of pseudo-metric, representative function, restriction window and any parameters of the above. In turn, the selection expectation depends on the structure of the probability space, whilst the Lebesgue measure criterion of the Vorob'ev expectation is not well suited for image analysis, although, the optimality of the Vorob'ev median is rather attractive. 

In our upcoming paper, we study the asymptotic properties of the ODF estimators for the boundary and describe a method for constructing confidence sets \citep{JankowskiStanberryXX}. 


\end{document}